\documentclass[12pt]{amsart}

\usepackage[margin=2.5cm]{geometry}
\usepackage{tikz}
\tikzset{node distance=2cm, auto}
\usetikzlibrary{matrix,arrows,decorations.pathmorphing,decorations.markings,calc}

\usepackage{graphicx}
\usepackage{caption}
\usepackage{subcaption}
\usepackage{amsmath,amssymb}

\usepackage{amscd}
\usepackage{verbatim}
\usepackage{enumerate}

\newtheorem{theorem}{Theorem}

\newtheorem{proposition}{Proposition}
\newtheorem{lemma}{Lemma}
\newtheorem{remark}{Remark}

\newtheorem{corollary}{Corollary}

\renewcommand{\H}{\mathcal{H}}

\newcommand{\RR}{\mathbb{R}}
\newcommand{\CC}{\mathbb{C}}
\newcommand{\ZZ}{\mathbb{Z}}

\newcommand{\pa}{\partial}

\newcommand{\be}{\begin{equation}}
\newcommand{\ee}{\end{equation}}

\newcommand{\x}{\times}

\title{Integrability of Limit Shapes of the Six Vertex Model}
\author{Nicolai Reshetikhin}
\address{N.R.: Department of Mathematics, University of California, Berkeley,
CA 94720, USA \\  \& KdV Institute for Mathematics, University of Amsterdam, 1098 XH Amsterdam, The Netherlands  \\ \& ITMO University. Saint Petersburg 197101, Russia.}
\email{reshetik@math.berkeley.edu}
\author{Ananth Sridhar}
\address{A.S.: Department of Physics, University of California, Berkeley,
CA 94720, USA}
\email{asridhar@berkeley.edu}
\begin{document}

\begin{abstract}
The main result of this paper is the construction of infinitely many conserved quantities (corresponding to commuting transfer-matrices) for the limit shape equation for the 6-vertex model on a cylinder. This suggests that the limit shape equation is an integrable PDE with gradient constraints. At the free fermionic point this equation becomes the complex Burgers equation.
\end{abstract}
\maketitle
\tableofcontents
\section*{Introduction}

For a large class of lattice models in equilibrium statistical mechanics,
in the thermodynamic limit, the system develops a deterministic component on the macroscopic scale, with statistical randomness remaining only at the microscopic scale.
This phenomenon is known as the limit shape phenomenon, for examples
and details see \cite{CKP}\cite{Sh}\cite{KOS}.
These types of limits are also known as hydrodynamical limits \cite{Land}.

In this paper, we study the limit shape equations for the 6-vertex model and dimer models
on a cylinder.  Both these models of statistical mechanics are "integrable" in some sense. The six vertex model is integrable in the sense that the weights of the model
satisfy the Yang-Baxter equation \cite{Ba}, have a number of special properties \cite{FadTakh}, and are closely
related to the representation theory of quantized universal algebra $U_q(\widehat{sl_2})$ \cite{JM}.
Dimer models are integrable in the sense that the partition functions and correlation functions can be
computed as Pfaffians of Kasteleyn matrices \cite{Ka}\cite{CR}.

Limit shapes with specific boundary conditions are well studied by now for dimer models, see for example
\cite{KO}\cite{OR}. For the 6-vertex model, limit shapes are best studied
for domain wall boundary conditions \cite{CoPr}\cite{CPZJ}.

The main result of this paper is the construction of infinitely many conserved quantities (integrals of motion)
for the nonlinear PDEs describing limit shapes in the 6-vertex and dimer models. We conjecture that these
equations describe an infinite dimensional integrable system. The existence of these integrals is closely related
to the commutativity of column-to-column transfer matrices.

The layout of the paper is as follows: In Section \ref{6v-section} we recall the 6-vertex model. The basics on the thermodynamic limit
are given in Section \ref{therm-lim-sect}. The Hamiltonian framework for limit shapes on a cylinder
is developed in Section \ref{Ham-6v-section}, where the commutativity of Hamiltonians is also proven. In
Section \ref{dimer-section}, we focus on the dimer model: the Hamiltonian framework is developed and Hamiltonians are explicitly computed. For dimer models, the limit shape equation is related to the complex Burgers equation. Section \ref{ff-point} contains the analysis of the free fermionic point $\Delta=0$ of the 6-vertex model, where
it can be reformulated as a dimer model. In the Conclusion, we discuss some open problems and conjectures.

A large part of this paper was completed while the authors visited Galileo Galilei Institute for Theoretical Physics.
The authors are grateful to GGI for the hospitality. N.R. also thanks Aarhus University
and the University of Tel Aviv where part of the work has been completed. The main result was announced
in 2014  Blumenthal lectures at University of Tel Aviv University. The work was supported by the
NSF grant DMS-0901431 and by the Chern-Simons endowment.

\section{The 6-vertex model}\label{6v-section}

\subsection{The partition function} \subsubsection{States and Weights}The six vertex model is most generally defined on a graph $ \Gamma = (V,E) $ with the property that each vertex has degree either one or four. We denote by $ \partial V $ the set of vertices of degree one and by $ \text{Int}(V) $ those of degree four, and denote by $ \partial E $ or $ \partial \Gamma $ the edges adjacent to a vertex in $ \partial V $ and by $ \text{Int}(E) = E \; \backslash \; \partial E $.

A state $ S $  of the six vertex model is a subset of edges (drawn on the graph as bold edges) satisfying the ice rule: each internal vertex has an even number of its adjacent edges contained in $ S $. There are six possible local configurations around each internal vertex $ v $, as shown in Figure \ref{fig:icerule}. The ice rule implies that a state of the six vertex model can be seen as a set of paths which do not cross (although they can touch at the $ w_2 $ vertex) and turn at vertices only in a particular way dictated by the ice rule. We assume that all edges of the 6-vertex model are oriented as shown.

\begin{figure}[h]
\begin{tikzpicture}  
\draw[decoration={markings,mark=at position 1 with {\arrow[scale=2]{>}}},
    postaction={decorate}] (-.8,0) -- (.8,0);
\draw[decoration={markings,mark=at position 1 with {\arrow[scale=2]{>}}},
    postaction={decorate}](0,-.8) -- (0,.8);
\node at (0,-1){$w_1$};
\end{tikzpicture} \; \;
\begin{tikzpicture}
\draw[decoration={markings,mark=at position 1 with {\arrow[scale=2]{>}}},
    postaction={decorate}] (-.8,0) -- (.8,0);
\draw[decoration={markings,mark=at position 1 with {\arrow[scale=2]{>}}},
    postaction={decorate}](0,-.8) -- (0,.8);
\draw[ultra thick](-.8,0)--(.76,0);
\draw[ultra thick](0,-.8)--(0,.76);
\node at (0,-1){$w_2$};
\end{tikzpicture} \; \;
\begin{tikzpicture}
\draw[decoration={markings,mark=at position 1 with {\arrow[scale=2]{>}}},
    postaction={decorate}] (-.8,0) -- (.8,0);
\draw[decoration={markings,mark=at position 1 with {\arrow[scale=2]{>}}},
    postaction={decorate}](0,-.8) -- (0,.8);
\draw[ultra thick](-.8,0)--(.76,0);
\node at (0,-1){$w_3$};
\end{tikzpicture} \; \;
\begin{tikzpicture}
\draw[decoration={markings,mark=at position 1 with {\arrow[scale=2]{>}}},
    postaction={decorate}] (-.8,0) -- (.8,0);
\draw[decoration={markings,mark=at position 1 with {\arrow[scale=2]{>}}},
    postaction={decorate}](0,-.8) -- (0,.8);
\draw[ultra thick](0,-.8)--(0,.76);
\node at (0,-1){$w_4$};
\end{tikzpicture} \; \;
\begin{tikzpicture}
\draw[decoration={markings,mark=at position 1 with {\arrow[scale=2]{>}}},
    postaction={decorate}] (-.8,0) -- (.8,0);
\draw[decoration={markings,mark=at position 1 with {\arrow[scale=2]{>}}},
    postaction={decorate}](0,-.8) -- (0,.8);
\draw[ultra thick](0,-.8)--(0,0)--(-.8,0);
\node at (0,-1){$w_5$};
\end{tikzpicture} \; \;
\begin{tikzpicture}
\draw[decoration={markings,mark=at position 1 with {\arrow[scale=2]{>}}},
    postaction={decorate}] (-.8,0) -- (.8,0);
\draw[decoration={markings,mark=at position 1 with {\arrow[scale=2]{>}}},
    postaction={decorate}](0,-.8) -- (0,.8);
\draw[ultra thick](.76,0)--(0,0)--(0,.76);
\node at (0,-1){$w_6$};
\end{tikzpicture}

\caption{Six vertex configurations and their weights.} \label{fig:icerule}
\end{figure}
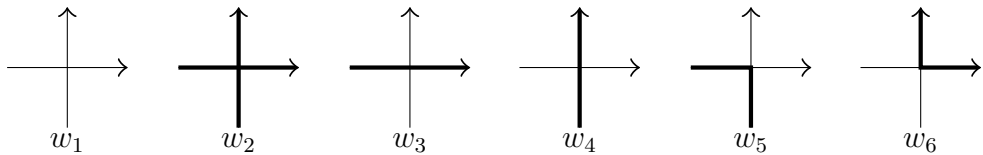

Each configuration of edges around a vertex is assigned a weight, which only depends on the state of edges adjacent to $ v $, and is given in Figure \ref{fig:icerule}.  The Boltzmann weight of the state $ W (S) $ is
\begin{align*}
W (S) = \prod_{ v \in \text{int}(V)} w(v,S)
\end{align*}
where $ w(v,S) $ is the weight assigned to vertex $ v $. 

\subsubsection{Boundary Conditions and the Partition Fuction}
The boundary state $ \partial S $ of $ S $ is the restriction of $S$ to boundary edges.
In other words, it is the subset of boundary edges contained in $ S $, ie. boundary edges occupied by paths.
Let $ \eta \subset \pa \Gamma$ be a subset of boundary edges. The partition given fixed boundary condition $ \eta $ is:
\begin{align*}
Z_{\Gamma,\eta}= \sum_{\substack{ \text{states}  \\ \partial S  = \eta}} W(S)
\end{align*}
Here we sum over all states with boundary state $\eta$.

\subsubsection{The Six Vertex Model on Planar Graphs, Cylinders, Torii}\label{sec:graphs}
From here on, we will focus in particular on the six vertex model defined on the following graphs:
\begin{enumerate}[a)]
\item Planar Domain Graph: These are planar graphs embedded in $ \mathbb{R}^2 $ such that each vertex has degree four or degree one, and each bounded face of the graph has precisely four vertices on its boundary.

\item Cylinder: Denote by $ \mathcal{C}_{MN}$ the square lattice graph embedded on the cylinder with $ N $ rows of vertices around the cylinder and $ M $ columns along its length (with edges along the boundary of the cylinder not included in the graph).

\item Torus: Denote by $ \mathcal{T}_{MN}$ the square lattice graph embedded on
the torus with $N$ rows and $M$ columns along the two cycles of the torus.
\end{enumerate}

These are examples of surface graphs. A surface graph is a graph $ \Gamma $ embedded in a compact, oriented surface $ \Sigma $ such that boundary vertices of $ \Gamma $ are embedded in $ \partial \Sigma $, and the complement of $ \Gamma $ in $ \Sigma $ consists of open two cells, called faces. In other words, $ \Gamma \cup \partial \Sigma $ is the 1-skeleton of a two dimensional cell decomposition of $ \Sigma $. 

For the planar and cylindrical graphs, we will say that a face of the surface graph is an inner face if it is bounded by edges of the graph. Otherwise, the face is called boundary face.

\begin{figure}[h]\label{fig:graphs}
\begin{tikzpicture}[scale=1]
\draw[step=.2 cm,very thin] (-1.3,-.7) grid (.3,.5);
\draw[step=.2 cm,very thin] (-.9,-1.1) grid (.3,1.1);
\draw[step=.2 cm,very thin] (.3,-1.1) grid (1.1,.1);
\end{tikzpicture}\; \; 
\begin{tikzpicture}[line join=round, scale = .8]
\draw(2.554,-.208)--(2.75,-.208);
\draw(2.554,.208)--(2.75,.208);
\draw(2.587,-.588)--(2.783,-.588);
\draw(2.587,.588)--(2.783,.588);
\draw(2.648,-.866)--(2.844,-.866);
\draw(2.648,.866)--(2.844,.866);
\draw(-.192,-.208)--(-.388,-.208);
\draw(-.192,.208)--(-.388,.208);
\draw(.196,0)--(.179,.407)--(.131,.743)--(.061,.951)--(-.02,.995)--(-.098,.866)--(-.159,.588)--(-.192,.208)--(-.192,-.208)--(-.159,-.588)--(-.098,-.866)--(-.02,-.995)--(.061,-.951)--(.131,-.743)--(.179,-.407)--(.196,0);
\draw(-.159,-.588)--(-.355,-.588);
\draw(-.159,.588)--(-.355,.588);
\draw(2.725,.995)--(2.921,.995);
\draw(2.725,-.995)--(2.921,-.995);
\filldraw[fill=white](2.725,-.995)--(2.806,-.951)--(2.463,-.951)--(2.382,-.995)--cycle;
\filldraw[fill=white](2.806,.951)--(2.725,.995)--(2.382,.995)--(2.463,.951)--cycle;
\filldraw[fill=white](2.382,-.995)--(2.463,-.951)--(2.12,-.951)--(2.039,-.995)--cycle;
\filldraw[fill=white](2.463,.951)--(2.382,.995)--(2.039,.995)--(2.12,.951)--cycle;
\filldraw[fill=white](2.039,-.995)--(2.12,-.951)--(1.777,-.951)--(1.696,-.995)--cycle;
\filldraw[fill=white](2.12,.951)--(2.039,.995)--(1.696,.995)--(1.777,.951)--cycle;
\draw(-.098,.866)--(-.294,.866);
\draw(-.098,-.866)--(-.294,-.866);
\filldraw[fill=white](1.696,-.995)--(1.777,-.951)--(1.433,-.951)--(1.352,-.995)--cycle;
\filldraw[fill=white](1.777,.951)--(1.696,.995)--(1.352,.995)--(1.433,.951)--cycle;
\filldraw[fill=white](1.352,-.995)--(1.433,-.951)--(1.09,-.951)--(1.009,-.995)--cycle;
\filldraw[fill=white](1.433,.951)--(1.352,.995)--(1.009,.995)--(1.09,.951)--cycle;
\filldraw[fill=white](1.009,-.995)--(1.09,-.951)--(.747,-.951)--(.666,-.995)--cycle;
\filldraw[fill=white](1.09,.951)--(1.009,.995)--(.666,.995)--(.747,.951)--cycle;
\draw(2.806,.951)--(3.002,.951);
\draw(2.806,-.951)--(3.002,-.951);
\filldraw[fill=white](2.806,-.951)--(2.877,-.743)--(2.534,-.743)--(2.463,-.951)--cycle;
\filldraw[fill=white](2.877,.743)--(2.806,.951)--(2.463,.951)--(2.534,.743)--cycle;
\filldraw[fill=white](.666,-.995)--(.747,-.951)--(.404,-.951)--(.323,-.995)--cycle;
\filldraw[fill=white](.747,.951)--(.666,.995)--(.323,.995)--(.404,.951)--cycle;
\filldraw[fill=white](2.534,.743)--(2.463,.951)--(2.12,.951)--(2.19,.743)--cycle;
\filldraw[fill=white](2.463,-.951)--(2.534,-.743)--(2.19,-.743)--(2.12,-.951)--cycle;
\filldraw[fill=white](.323,-.995)--(.404,-.951)--(.061,-.951)--(-.02,-.995)--cycle;
\filldraw[fill=white](.404,.951)--(.323,.995)--(-.02,.995)--(.061,.951)--cycle;
\filldraw[fill=white](2.19,.743)--(2.12,.951)--(1.777,.951)--(1.847,.743)--cycle;
\filldraw[fill=white](2.12,-.951)--(2.19,-.743)--(1.847,-.743)--(1.777,-.951)--cycle;
\draw(-.02,.995)--(-.217,.995);
\draw(-.02,-.995)--(-.217,-.995);
\filldraw[fill=white](1.777,-.951)--(1.847,-.743)--(1.504,-.743)--(1.433,-.951)--cycle;
\filldraw[fill=white](1.847,.743)--(1.777,.951)--(1.433,.951)--(1.504,.743)--cycle;
\filldraw[fill=white](1.433,-.951)--(1.504,-.743)--(1.161,-.743)--(1.09,-.951)--cycle;
\filldraw[fill=white](1.504,.743)--(1.433,.951)--(1.09,.951)--(1.161,.743)--cycle;
\draw(2.877,.743)--(3.073,.743);
\draw(2.877,-.743)--(3.073,-.743);
\filldraw[fill=white](1.09,-.951)--(1.161,-.743)--(.818,-.743)--(.747,-.951)--cycle;
\filldraw[fill=white](1.161,.743)--(1.09,.951)--(.747,.951)--(.818,.743)--cycle;
\filldraw[fill=white](2.925,.407)--(2.877,.743)--(2.534,.743)--(2.582,.407)--cycle;
\filldraw[fill=white](2.877,-.743)--(2.925,-.407)--(2.582,-.407)--(2.534,-.743)--cycle;
\filldraw[fill=white](.747,-.951)--(.818,-.743)--(.474,-.743)--(.404,-.951)--cycle;
\filldraw[fill=white](.818,.743)--(.747,.951)--(.404,.951)--(.474,.743)--cycle;
\filldraw[fill=white](2.582,.407)--(2.534,.743)--(2.19,.743)--(2.238,.407)--cycle;
\filldraw[fill=white](2.534,-.743)--(2.582,-.407)--(2.238,-.407)--(2.19,-.743)--cycle;
\filldraw[fill=white](.404,-.951)--(.474,-.743)--(.131,-.743)--(.061,-.951)--cycle;
\filldraw[fill=white](.474,.743)--(.404,.951)--(.061,.951)--(.131,.743)--cycle;
\filldraw[fill=white](2.238,.407)--(2.19,.743)--(1.847,.743)--(1.895,.407)--cycle;
\filldraw[fill=white](2.19,-.743)--(2.238,-.407)--(1.895,-.407)--(1.847,-.743)--cycle;
\draw(.061,.951)--(-.136,.951);
\draw(.061,-.951)--(-.136,-.951);
\draw(2.925,.407)--(3.121,.407);
\draw(2.925,-.407)--(3.121,-.407);
\filldraw[fill=white](1.895,.407)--(1.847,.743)--(1.504,.743)--(1.552,.407)--cycle;
\filldraw[fill=white](1.847,-.743)--(1.895,-.407)--(1.552,-.407)--(1.504,-.743)--cycle;
\filldraw[fill=white](2.942,0)--(2.925,.407)--(2.582,.407)--(2.599,0)--cycle;
\filldraw[fill=white](2.925,-.407)--(2.942,0)--(2.599,0)--(2.582,-.407)--cycle;
\filldraw[fill=white](1.504,-.743)--(1.552,-.407)--(1.209,-.407)--(1.161,-.743)--cycle;
\filldraw[fill=white](1.552,.407)--(1.504,.743)--(1.161,.743)--(1.209,.407)--cycle;
\draw(2.942,0)--(3.138,0);
\filldraw[fill=white](2.599,0)--(2.582,.407)--(2.238,.407)--(2.255,0)--cycle;
\filldraw[fill=white](2.582,-.407)--(2.599,0)--(2.255,0)--(2.238,-.407)--cycle;
\filldraw[fill=white](1.161,-.743)--(1.209,-.407)--(.866,-.407)--(.818,-.743)--cycle;
\filldraw[fill=white](1.209,.407)--(1.161,.743)--(.818,.743)--(.866,.407)--cycle;
\filldraw[fill=white](2.255,0)--(2.238,.407)--(1.895,.407)--(1.912,0)--cycle;
\filldraw[fill=white](2.238,-.407)--(2.255,0)--(1.912,0)--(1.895,-.407)--cycle;
\filldraw[fill=white](.818,-.743)--(.866,-.407)--(.522,-.407)--(.474,-.743)--cycle;
\filldraw[fill=white](.866,.407)--(.818,.743)--(.474,.743)--(.522,.407)--cycle;
\filldraw[fill=white](1.912,0)--(1.895,.407)--(1.552,.407)--(1.569,0)--cycle;
\filldraw[fill=white](1.895,-.407)--(1.912,0)--(1.569,0)--(1.552,-.407)--cycle;
\filldraw[fill=white](.474,-.743)--(.522,-.407)--(.179,-.407)--(.131,-.743)--cycle;
\filldraw[fill=white](.522,.407)--(.474,.743)--(.131,.743)--(.179,.407)--cycle;
\filldraw[fill=white](1.569,0)--(1.552,.407)--(1.209,.407)--(1.226,0)--cycle;
\filldraw[fill=white](1.552,-.407)--(1.569,0)--(1.226,0)--(1.209,-.407)--cycle;
\draw(.131,.743)--(-.065,.743);
\draw(.131,-.743)--(-.065,-.743);
\filldraw[fill=white](1.209,-.407)--(1.226,0)--(.883,0)--(.866,-.407)--cycle;
\filldraw[fill=white](1.226,0)--(1.209,.407)--(.866,.407)--(.883,0)--cycle;
\filldraw[fill=white](.866,-.407)--(.883,0)--(.539,0)--(.522,-.407)--cycle;
\filldraw[fill=white](.883,0)--(.866,.407)--(.522,.407)--(.539,0)--cycle;
\filldraw[fill=white](.522,-.407)--(.539,0)--(.196,0)--(.179,-.407)--cycle;
\filldraw[fill=white](.539,0)--(.522,.407)--(.179,.407)--(.196,0)--cycle;
\draw(.179,.407)--(-.017,.407);
\draw(.179,-.407)--(-.017,-.407);
\draw(.196,0)--(0,0);
\end{tikzpicture}
\; \;
\input{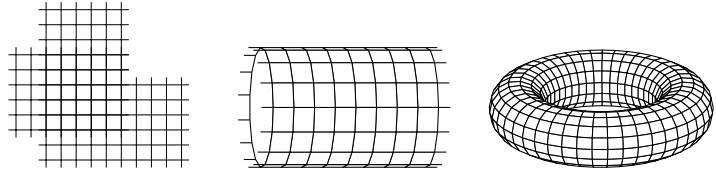}
\caption{A planar graph, cylinder graph, and torus graph.}
\end{figure}

\subsubsection{Magnetic Fields}
One set of weights is particularly important in physical application. It is given in terms of $ (a,b,c) $ and magnetic field $ (H,V) $ as:
\begin{align*}
\begin{array}{cc}
 w_1 = a \; e^{H+V} \; \;  & w_2 = a \; e^{-H - V} \\
 w_3 = b \; e^{H-V}  \; \; & w_4= b \; e^{-H+V} \\
 w_5 = c \; \;   &  w_6 = c
\end{array}
\end{align*}

These weights have the following interpretation: vertices are assigned the weights $ w_1 = w_2 =  a; \; \; w_3 = w_4 = b; \; \; w_5 = w_6 = c $. The vertex weight (depending only on $a,b,c$) of a configuration is locally invariant when $ S $ is replaced by $ \Gamma \backslash S $, i.e. by interchanging occupied and unoccupied edges. This symmetry is broken by assigning each occupied horizontal edge a weight $ e^{H/2} $ and each unoccupied edge $ e^{-H/2} $. Similarly we assign $ e^{V/2} $ and $ e^{-V/2} $ to occupied and unoccupied vertical edges. The total weight of a state is then the product of weights on edges and weights on vertices.

An important characteristic of the 6-vertex model weight is:
\begin{equation}
\Delta = \frac{ w_1 w_2 + w_3 w_4 - w_5 w_6}{2 \sqrt{ w_1 w_2 w_3 w_4}} \ \
= \frac{a^2 + b^2 - c^2}{2 a b}
\end{equation}

\subsection{Commuting transfer matrices}

\subsubsection{Baxter's parametrization of weights} \label{subsec:baxter}
A useful parametrization of weights $ (a,b,c) $ in terms of $ (u, \gamma, r) $ was introducted by Baxter \cite{Ba}:
\begin{enumerate}[a)]
\item When $ \Delta > 1 $:
\begin{enumerate}[i]
\item If $ a > b+ c $, let $ (a, b, c) =  (r \sinh(u + \gamma), r \sinh(u), r \sinh(\gamma) ) $ with $ \gamma > 0 $.
\item If $ b> a+c $, let $ (a,b,c) =  (r \sinh(u - \gamma), r \sinh(u), r \sinh(\gamma) ) $ with $ 0 < \gamma < u $.
\end{enumerate}
For both of these parametrizations $ \Delta = \cosh(\gamma) $.

\item When $ | \Delta | < 1 $:
\begin{enumerate}[i]
\item If $ -1 <  \Delta \leq 0 $, let $ (a, b,  c) = (r \sin(u - \gamma),  r \sin(u),  r \sin(\gamma) ) $  with $ 0 < \gamma < \pi/2 $ and $ \gamma < u < \pi/2 $.
\item If $ 0 <  \Delta \leq 1 $, let $ (a, b, c) = (r \sin(\gamma - u), r \sin(u), r \sin(\gamma) ) $ with $ 0 < \gamma < \pi/2 $ and $ 0< u < \gamma $.
\end{enumerate}
For these parametrizations $ \Delta = \pm \cos( \gamma) $.

\item When $ \Delta < -1 $: let $ (a,b,c) = ( r \sinh(\gamma - u ) , r \sinh(u), r \sinh( \gamma )) $ with  $ 0 < u < \gamma $. \\
In this case $ \Delta = - \cosh( \gamma) $.
\end{enumerate}
We will refer to the variable $u$ as the spectral parameter.

\subsubsection{The $ R$-matrix and the Yang-Baxter equation}It is convenient to arrange the weights of the 6-vertex model into a
$ 4 \times 4 $ matrix as follows:
\begin{align} \label{eq:rmatrix}
R =
\begin{pmatrix}
a \; e^{H+V} & 0 & 0 & 0 \\
0 & b \; e^{H-V} & c  & 0 \\
0 & c & b\; e^{V-H} & 0 \\
0 & 0  & 0 & a \; e^{-H-V}
\end{pmatrix}
\end{align}
Let $ e_1 $ and $ e_2 $ be the standard basis of $\mathbb{C}^2 $:
\begin{align*}
  e_1 = \begin{pmatrix} 1 \\ 0 \end{pmatrix} \hspace{15pt} e_2 = \begin{pmatrix} 0 \\ 1 \end{pmatrix}
\end{align*}
The state of each edge corresponds to a vector in $ \mathbb{C}^2 $, with an empty edge corresponding to $ e_1 $, and an occupied edge to $ e_2 $.  In the tensor product basis, $e_1\otimes e_1, e_1\otimes e_2,  e_2\otimes e_1,  e_2\otimes e_2 $, the R-matrix (\ref{eq:rmatrix}) represents a linear operator that maps a state on the West and South edges adjacent to a vertex to a state on the North and East edges, scaled by the Boltzmann weight.

Denote by $ R(u,H,V) $ the $ R $-matrix (\ref{eq:rmatrix}) with weights given by the parametrization (\ref{subsec:baxter}). Let $ R(u) = R(u,0,0) $. Then the $ R$-matrices satisfy the Yang-Baxter equation \cite{Ba}:
\begin{align} \label{eq:ybe}
R_{12}(u) R_{13}(u+v) R_{23}(v) = R_{23}(v) R_{13}(u+v) R_{12}(u)
\end{align}
in $ \mathbb{C}^2 \otimes \mathbb{C}^2 \otimes \mathbb{C}^2 $. Here the parameter $\Delta$
is the same for all three matrices.

The $ R $-matrix with magnetic field can be expressed as:
\begin{align} \label{eq:brmatrix}
R(u,H,V) = (D^H \otimes D^V) \; R(u) \; (D^H \otimes D^V)
\end{align}
where
\begin{align*}
D^H = \begin{pmatrix}
e^{H/2} & 0 \\
0 & e^{-H/2}
\end{pmatrix}
\end{align*}

The ice-rule implies that for any diagonal matrix $ D $:
\begin{align}  \label{eq:diagrmatrix}
 (D \otimes D) \; R(u, H,V) = R(u,H,V) \; (D \otimes D)
\end{align}
Together (\ref{eq:diagrmatrix}) and (\ref{eq:brmatrix}) imply that $R(u, H, V)$ satisfies:
\begin{equation}  \label{eq:ybe2}
R_{12}(u) R_{13}(u+v,0,V) R_{23}(v,0,V) =  R_{23}(v,0,V) R_{13}(u+v,0,V) R_{12}(u)
\end{equation}

\subsubsection{Commutativity of column-to-column transfer matrix}
Let $V$ be the space of states on $N $ horizontal edges in a slice between two columns of the square lattice:
\begin{align*}
V &=   \mathbb{C}^2 \otimes  \cdots \otimes \mathbb{C}^2 = {\CC^2}^{\otimes N}
\end{align*}
The quantum monodromy matrix is the linear operator $ T_a:\mathbb{C}^2 \otimes V \rightarrow \mathbb{C}^2 \otimes V $ defined by:
\begin{align} \label{eq:mon}
T_a(u, H, V)= R_{1a}(u,H,V) R_{2a}(u,H,V) \cdots R_{Na}(u,H,V)
\end{align}
Here we have enumerated factors in the tensor product $ \mathbb{C}^2 \otimes V = \mathbb{C}^2 \otimes \mathbb{C}^2 \otimes \cdots \otimes \mathbb{C}^2 $ by $ ( a,1, \cdots, N ) $. The operator $ R_{ai} $ acts on $ \mathbb{C}^2 \otimes V $ as $ R $ on the product $ \mathbb{C}^2 \otimes \mathbb{C}^2 $ of factors enumerated by $a $  and $ i $.
The matrix elements of the quantum monodromy matrix are partition functions (with different boundary conditions) for a planar domain with one column of $ N$ vertices.

The transfer matrix $ t(u,H,V): V \rightarrow V $  is defined as the partial trace over the first factor in $ \mathbb{C}^2 \otimes V $:
\begin{align} \label{eq:transfermatrix}
t(u, H, V) = ( \text{Tr}\otimes \text{Id} ) \; T_a(u, H, V) 
\end{align}

We identify a boundary state $\eta$ on one end of the cylinder $ \mathcal{C}_{MN}$ with a vector
\begin{align*}
\psi_\eta = e_{\eta(1) } \otimes e_{\eta(2)} \cdots \otimes e_{\eta(N) } \in V
\end{align*}
where $ i = 1, \cdots, N $ enumerate boundary edges, and $ \eta(i) = 1,2 $;  $ e_1 $ corresponds to an empty edge, and $ e_2 $ to an edge occupied by a path. The set of vectors $\psi_\eta$ taken over all $ \eta $ form a basis in $V$.

It is clear from the definition of the transfer matrix that
\begin{equation}\label{ctm1}
t(u,H,V)=(D_H\otimes \dots \otimes D_H)t(u, 0, V)(D_H\otimes \dots \otimes D_H)
\end{equation}
The 6-vertex rule implies that
\begin{equation}\label{ctm2}
(D_A\otimes \dots \otimes D_A)t(u,H, V)= t(u,H, V)(D_A\otimes \dots \otimes D_A)
\end{equation}
for any $A$.

The Yang-Baxter equation (\ref{eq:ybe2}) implies that
\[
T_a(u,0, V)T_b(u,0,V)R_{ab}(u)=R_{ab}(u)T_b(u,0,V)T_a(u,0, V)
\]
Together with (\ref{ctm1}) and (\ref{ctm2}) this implies the commutativity of transfer matrices
\[
[t(u, H, V), t(v, H,V)]=0
\]

The partition function for the six vertex model on the cylinder with fixed boundary condition is a matrix element of the product of transfer matrices:
\begin{align*}
Z_{\mathcal{C}_{MN}, \eta_1, \eta_2 } = ( \psi_{\eta_1}, t(u,H,V)^M\psi_{\eta_2})
\end{align*}
Here $(x,y)$ is the natural scalar product on ${\CC^2}^{\otimes N}$. Let $m(\eta)$ be the difference
between the number of horizontal edges occupies by paths and the number of non-occupied edges.
The 6-vertex rule implies that $m(\eta_1)=m(\eta_2)$ and
\begin{equation}\label{cyl-mag-field}
Z_{\mathcal{C}_{MN}, \eta_1, \eta_2 } = ( \psi_{\eta_1}, t(u,0,V)^M\psi_{\eta_2})\exp \big(H m(\eta_1) / 2 \big)
\end{equation}
The partition function of the six vertex model on the torus $ \mathcal{T}_{MN} $ is:
\begin{align*}
Z_{\mathcal{T}_{MN}} = \text{Tr} \; t^{M} =\sum_\eta Z_{\mathcal{C}_{MN}, \eta, \eta }
\end{align*}

\subsection{Height Functions}
We say a map $\theta$ from faces of $\Gamma$ to the half integers is a {\it height function} if at every vertex it satisfies the ice rule for height functions (see Figure \ref{fig:iceruleheights}). Such functions are locally in bijection with $6$-vertex configurations.

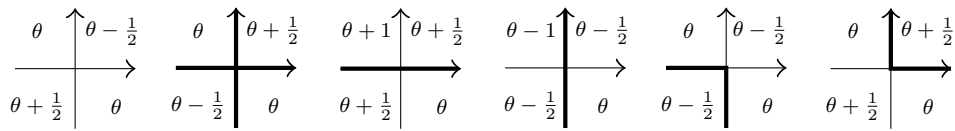
\begin{figure}[h] \label{fig:iceruleheights}
\begin{tikzpicture}
\draw[decoration={markings,mark=at position 1 with {\arrow[scale=2]{>}}},
    postaction={decorate}] (-.8,0) -- (.8,0);
\draw[decoration={markings,mark=at position 1 with {\arrow[scale=2]{>}}},
    postaction={decorate}](0,-.8) -- (0,.8);
\node at (.5,-.5){\tiny{ $\theta$ }};
\node at (.5,.5){\tiny{$\theta - \frac{1}{2} $  }};
\node at (-.5,-.5){ \tiny{$\theta + \frac{1}{2}$} };
\node at (-.5,.5){ \tiny{$\theta$} };
\end{tikzpicture} 
\begin{tikzpicture}
\draw[decoration={markings,mark=at position 1 with {\arrow[scale=2]{>}}},
    postaction={decorate}] (-.8,0) -- (.8,0);
\draw[decoration={markings,mark=at position 1 with {\arrow[scale=2]{>}}},
    postaction={decorate}](0,-.8) -- (0,.8);
\draw[ultra thick](-.8,0)--(.76,0);
\draw[ultra thick](0,-.8)--(0,.76);
\node at (.5,-.5){ \tiny{$\theta$}};
\node at (.5,.5){\tiny{$\theta + \frac{1}{2} $  }};
\node at (-.5,-.5){ \tiny{$\theta - \frac{1}{2}$ }};
\node at (-.5,.5){\tiny{ $\theta$ } };
\end{tikzpicture} 
\begin{tikzpicture}
\draw[decoration={markings,mark=at position 1 with {\arrow[scale=2]{>}}},
    postaction={decorate}] (-.8,0) -- (.8,0);
\draw[decoration={markings,mark=at position 1 with {\arrow[scale=2]{>}}},
    postaction={decorate}](0,-.8) -- (0,.8);
\draw[ultra thick](-.8,0)--(.76,0);
\node at (.5,-.5){\tiny{ $\theta$ }};
\node at (.5,.5){\tiny{$\theta + \frac{1}{2} $  }};
\node at (-.5,-.5){\tiny{ $\theta + \frac{1}{2}$} };
\node at (-.5,.5){\tiny{ $\theta + 1$ }};
\end{tikzpicture}
\begin{tikzpicture}
\draw[decoration={markings,mark=at position 1 with {\arrow[scale=2]{>}}},
    postaction={decorate}] (-.8,0) -- (.8,0);
\draw[decoration={markings,mark=at position 1 with {\arrow[scale=2]{>}}},
    postaction={decorate}](0,-.8) -- (0,.8);
\draw[ultra thick](0,-.8)--(0,.76);
\node at (.5,-.5){ \tiny{$\theta$ }};
\node at (.5,.5){\tiny{$\theta - \frac{1}{2} $  }};
\node at (-.5,-.5){\tiny{ $\theta - \frac{1}{2}$ }};
\node at (-.5,.5){\tiny{ $\theta - 1$} };
\end{tikzpicture} 
\begin{tikzpicture}
\draw[decoration={markings,mark=at position 1 with {\arrow[scale=2]{>}}},
    postaction={decorate}] (-.8,0) -- (.8,0);
\draw[decoration={markings,mark=at position 1 with {\arrow[scale=2]{>}}},
    postaction={decorate}](0,-.8) -- (0,.8);
\draw[ultra thick](0,-.8)--(0,0)--(-.8,0);
\node at (.5,-.5){\tiny{ $\theta$} };
\node at (.5,.5){\tiny{$\theta - \frac{1}{2} $ } };
\node at (-.5,-.5){ \tiny{$\theta - \frac{1}{2}$ }};
\node at (-.5,.5){ \tiny{$\theta $ }};
\end{tikzpicture} 
\begin{tikzpicture}
\draw[decoration={markings,mark=at position 1 with {\arrow[scale=2]{>}}},
    postaction={decorate}] (-.8,0) -- (.8,0);
\draw[decoration={markings,mark=at position 1 with {\arrow[scale=2]{>}}},
    postaction={decorate}](0,-.8) -- (0,.8);
\draw[ultra thick](.8,0)--(0,0)--(0,.76);
\node at (.5,-.5){ \tiny{$\theta$ }};
\node at (.5,.5){\tiny{$\theta + \frac{1}{2} $}  };
\node at (-.5,-.5){ \tiny{$\theta + \frac{1}{2}$} };
\node at (-.5,.5){ \tiny{$\theta $ }};
\end{tikzpicture}
 
\caption{The ice rule for height functions and the corresponding six vertex states.}\label{fig:iceruleheights}
\end{figure}

For a simply connected domain this local bijection gives a global bijection between
ice configurations and height functions modulo constants (or, equivalently,
height functions with a fixed value at a chosen reference face).

\begin{remark} A more familiar height function $ \xi $ is defined so that its level curves are the paths of the six vertex state, see for example \cite{AR}. The height function defined here differs by a linear function $\xi=\theta+x/2+y/2$.
\end{remark}

If the surface graph is not simply connected, it is clear that height functions may not exist globally.
For the cylinder, a height function can instead be regarded as a multivalued function (a section of a line bundle) with the non-trivial monodromy across the cylinder. The ice rule guarantees that the monodromy does not depend on the choice of a cycle as long as it is homotopic to a simple curve across the cylinder.

Similarly, for the torus, if we choose a basis of cycles, say $a$ and $b$ , a height function gains monodromy along each of the cycles, which, because of the ice rule depends only on the homotopy class of a cycle
on $\Gamma$.

To make all height functions globally defined as functions on faces, we will choose branch cuts for cylinders and tori.
Denote by $H_{\widetilde{C}_{MN}}(\eta_1 , \eta_2 )$ the space of all height functions on the cylinder $\widetilde{C}_{MN}$ with
a branch cut along one of the grid lines and with boundary height functions corresponding to $\eta_1$ and $\eta_2$.

Similarly we will denote by $H_{\widetilde{T}_{MN}}(\Delta_x \theta, \Delta_y \theta)$ the space of all height functions on the torus $\widetilde{T}_{MN}$ with branch cuts along $a$ and $b$, with the monodromy $\Delta_y \theta$ across $b$ and  monodromy $\Delta_x \theta$ across $a$.

Because states are defined also on boundary edges, the height function is defined
not only on inner faces of the grid, but also on the outer faces. The boundary height function is defined
by a boundary configuration, up to a constant on each connected component of the boundary.

\begin{remark} One can regard a six vertex state as a $1$-cycle on the $2$-dimensional
complex of the surface graph in question. In the case when the complex is connected and simply connected, a cycle is the boundary of a $2$-cycle, which is the height function. But because
our interest is mostly the case of a cylinder, we will use somewhat cumbersome notion
of a height function on a surface with branch cuts.
\end{remark}

\section{The thermodynamic limit and limit shapes}\label{therm-lim-sect}
In this section, we consider approximating a planar domain, cylinder, or torus by the square lattice graphs, and study the six vertex model on these graphs as the mesh $ \epsilon $ of the approximation approaches zero. In this limit, there is the emergence of the limit shape phenomena: the rescaled height functions has vanishing variance and converges to a surface called the limit shape that can be found by solving a certain boundary value problem.

\subsection{Embedded Surfaces Graphs}\label{sec:approx}
Here we define families of graphs that in the limit of zero mesh fill the corresponding surfaces.

\begin{enumerate}[a)]
\item Planar domain: Let $D$ be a simply connected domain in $ \mathbb{R}^2 $ with Euclidean metric
and $\phi_\epsilon: \ZZ^2\to \RR^2$ the embedding given by $ (m,n) \rightarrow (\epsilon m, \epsilon n) $. For generic domain $ D $, the intersection  $ D$ with $\phi_\epsilon(\ZZ^2)$ defines a planar domain graph, which we call $ D_\epsilon $. (See Figure \ref{fig:domain2}).

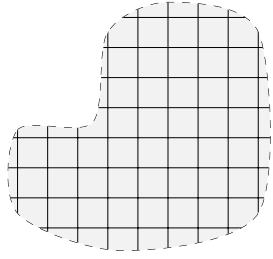
\begin{figure}[h] \label{fig:domain2}
\begin{tikzpicture}[scale=2]
\draw [dashed] plot [smooth cycle] coordinates {(-1,-.5) (-1,.05) (-.5,.1) (-.4,.7)(0,.9)(.6,.7)(.6,-.5)(-.3,-.75)};
\fill [color= black!5!white] plot [smooth cycle] coordinates {(-1,-.5) (-1,.05) (-.5,.1) (-.4,.7)(0,.9)(.6,.7)(.6,-.5)(-.3,-.75)};

\clip  plot [smooth cycle] coordinates {(-1,-.5) (-1,.05) (-.5,.1) (-.4,.7)(0,.9)(.6,.7)(.6,-.5)(-.3,-.75)};

\draw[step=.2 cm,very thin] (-1.5,-1) grid (1.5,1);
\end{tikzpicture}

\caption{Approximating a domain $ D $ with a graph $ D_\epsilon $.} \label{fig:domain2}
\end{figure}

\item Cylinder: Let $C_{TL} =   [0,T] \times \RR \; / \;   \{ (x,y) \sim (x,y+L) \}  $  be the cylinder of length $T$,  circumference $L$ and with flat Euclidean metric. We take the branch cut along $ y = 0$. Denote by $C_{TL}^{(\epsilon)}$ the intersection of $\phi_\epsilon(\ZZ^2)$ with the fundamental domain $  [0,T] \times \RR $.

\item  Torus: Let $ T_{TL} = \RR \times \RR \; / \; \{ (x,y) \sim (x,y+L) \sim (x+T,y) \} $ be the flat torus with constant Euclidean metric. We take the branch cut along $ x =0 $ and $ y = 0 $. Denote by $T_{TL}^{(\epsilon)}$ the intersection of $\phi_\epsilon(\ZZ^2)$  with the fundamental domain $ [0,T] \times [0,L] $.
\end{enumerate}

\begin{figure}[h] \label{fig:toruscylinderbranch}
\begin{tikzpicture}[baseline,scale = .3]
\fill (-4,0) circle (0pt);
\fill (4,0) circle (0pt) ;
\fill (0,-5) circle (0pt) ;

\draw [thick, domain = 0:360,smooth] plot ({cos(\x)*(3 + 1/sqrt(1 + 3.35069*sin(\x)^2))}, {1.6064*sin(\x) / sqrt(1 + 3.35069*sin(\x)^2 ) + .479426*sin(\x)*( 3 + 1/sqrt( 1 + 3.3069*sin(\x)^2 ) ) } );

\draw [thick, domain = 17:163,smooth] plot ({-cos(\x)*(3 - 1/sqrt(1 + 3.35069*sin(\x)^2))}, {1.6064*sin(\x) / sqrt(1 + 3.35069*sin(\x)^2 ) - .479426*sin(\x)*( 3 - 1/sqrt( 1 + 3.3069*sin(\x)^2 ) ) } );

\draw [thick, domain = 35:145, smooth] plot ({cos(\x)*(3 - 1/sqrt(1 + 3.35069*sin(\x)^2))}, {-1.6064*sin(\x) / sqrt(1 + 3.35069*sin(\x)^2 ) + .479426*sin(\x)*( 3 - 1/sqrt( 1 + 3.3069*sin(\x)^2 ) ) } );

\draw [domain = 180:360, smooth] plot ({ 4*cos(\x) }, {1.9177*sin(\x) } );

\draw [dashed, domain = 0:180, smooth] plot ({ 4*cos(\x) }, {1.9177*sin(\x) } );

\draw [decoration={markings,mark=at position .5 with {\arrow[scale=2]{>>}}},
    postaction={decorate}] (-.01,-1.9177)--(.01,-1.9177);

\draw [decoration={markings,mark=at position .5 with {\arrow[scale=2]{>}}},
    postaction={decorate}] (-1.93767, -1.38379) -- (-1.93692, -1.37897);

\draw [domain = -47:122, smooth] plot ({-0.490261*(3 + cos(\x)) }, {-0.417856*(3 + cos(\x)) + 0.877583*sin(\x) } );

\draw [dashed, domain = 122:360, smooth] plot ({-0.490261*(3 + cos(\x)) }, {-0.417856*(3 + cos(\x)) + 0.877583*sin(\x) } );
\end{tikzpicture} \;
\begin{tikzpicture}[baseline, scale = .55]
\draw (-2,-1.5) rectangle (2,1.5);
\draw [decoration={markings,mark=at position .5 with {\arrow[scale=2]{>>}}},
    postaction={decorate}] (-2,-1.5)--(2,-1.5);
\draw [decoration={markings,mark=at position .5 with {\arrow[scale=2]{>>}}},
    postaction={decorate}] (-2,1.5)--(2,1.5);

\draw [decoration={markings,mark=at position .5 with {\arrow[scale=2]{>}}},
    postaction={decorate}] (-2,-1.5)--(-2,1.5);
\draw [decoration={markings,mark=at position .5 with {\arrow[scale=2]{>}}},
    postaction={decorate}] (2,-1.5)--(2,1.5);

\node at  (-2.3,-1.8) {\footnotesize{(0,0)}};
\node at  (2.3,-1.8) { \footnotesize{(T,0)} };
\node at  (-2.3,1.8) { \footnotesize{ (0,L)} };
\node at  (2.3,1.8) { \footnotesize{ (T,L)} };
\end{tikzpicture} \;\; \; \; \; \;
\begin{tikzpicture}[baseline,scale = .3]
\fill (-4,0) circle (0pt);
\fill (4,0) circle (0pt) ;
\draw [domain = 0:360, smooth] plot ({3+ .6*cos(\x)}, { 2*sin(\x) } );
\draw [domain = 90:270, smooth] plot ({-3 + .6*cos(\x)}, { 2*sin(\x) } );
\draw (3,2)--(-3,2);
\draw (3,-2)--(-3,-2);
\draw [decoration={markings,mark=at position .5 with {\arrow[scale=2]{>}}},
    postaction={decorate}] (-3.5732,0.59104)-- (2.4268,0.59104);
\end{tikzpicture} \;
\begin{tikzpicture}[baseline, scale = .55]
\draw (-2,-1.5) rectangle (2,1.5);
\draw [decoration={markings,mark=at position .5 with {\arrow[scale=2]{>}}},
    postaction={decorate}] (-2,-1.5)--(2,-1.5);
\draw [decoration={markings,mark=at position .5 with {\arrow[scale=2]{>}}},
    postaction={decorate}] (-2,1.5)--(2,1.5);

\draw (-2,-1.5)--(-2,1.5);
\draw (2,-1.5)--(2,1.5);

\node at  (-2.3,-1.8) {\footnotesize{(0,0)}};
\node at  (2.3,-1.8) { \footnotesize{(T,0)} };
\node at  (-2.3,1.8) { \footnotesize{ (0,L)} };
\node at  (2.3,1.8) { \footnotesize{ (T,L)} };
\end{tikzpicture}
 
\caption{Branch cuts on the torus and cylinder.}
\end{figure}
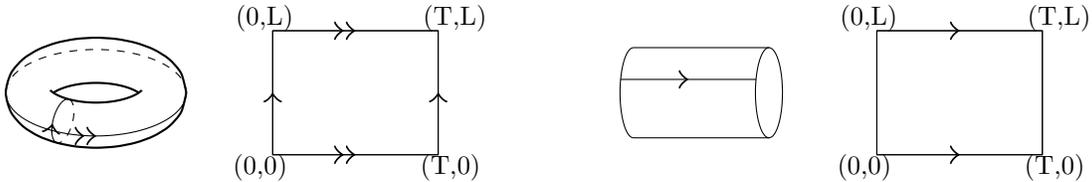

\subsection{Normalized Height Function and Boundary Height Function}
For a domain $D$ a {\it normalized height function} is defined as a piece-wise constant
function, with constant value $ h^{\text{norm}} = \epsilon \theta(f) $ on each face $ f \in \Gamma \subset D $. As a consequence of the ice rule, a normalized height function satisfies
\begin{align*}
 -\epsilon/2 \leq h^{\text{norm}}(x +\epsilon,y) - h^{\text{norm}}(x  ,y) \leq \epsilon/2\\
 -\epsilon/2 \leq h^{\text{norm}}(x, y+\epsilon) - h^{\text{norm}}(x,y) \leq \epsilon/2
\end{align*}

A normalized {\it boundary height function} is a piece-wise constant function on $\pa D$. Its value on
the segment of $\pa D$ which intersects an outer face $f$ of $D_\epsilon$ is $\epsilon \theta(f)$. Normalized boundary height function changes by $\pm \epsilon$ or does not change at each intersection point of $\pa D$ and $\phi_\epsilon(\ZZ^2)$.
The sign is determined by the orientation of lines in the square grid and by the orientation of $\RR^2$.

Similarly we define normalized height functions for a cylinder and a torus with branch cuts.

\subsection{The Thermodynamic Limit on a Torus}Let $ \mathbb{T} = T_{11} $ be the unit torus, and $ \{ M^{(n)}\}_{(n = 1 \cdots \infty)} $ be a sequence of embedded surface graphs with mesh $\epsilon_n$ as defined above, approximating $ M $  as $ \epsilon_n \rightarrow 0 $.

It is expected\footnote{The evidence for this is overwhelming, though
it is hard to point a specific reference where it has been proven with
mathematical level of rigor.} that as $n\to \infty$ there exist the limiting
density of the free energy
\begin{equation}\label{free-en-t}
f_{\mathbb{T}} = \lim_{n\rightarrow \infty} \epsilon_n^2 \log(Z_{M^{(n)}}) ,
\end{equation}

We will call $f_{\mathbb{T}}(H,V)$ the {\it toroidal free energy}. Its Legendre transform $\sigma(s,t)$ is the free energy
for torus with zero magnetic fields conditioned to having the average magnetization $s$ in the horizontal direction and $t$ in the vertical direction:
\begin{align}
\sigma(s,t) = \max_{H,V} (s H + t V - f_\mathbb{T}(H,V)), \ \ s,t\in [-\frac{1}{2},\frac{1}{2}]
\end{align}

Properties of the toroidal free energy $ f_{\mathbb{T}}(H,V) $ and of $ \sigma(s,t) $, are summarized in \cite{PR}
(see also references therein). We briefly recall here some of them:

First, it is clear that $ \sigma $ has the following symmetries:
\begin{align*}
\sigma(s,t) = \sigma(t,s) = \sigma(- s, - t)
\end{align*}

From which we also have $ \frac{\partial \sigma}{\partial s}(s,t) =  \frac{\partial \sigma}{\partial t}(t,s) $.

In addition, $ \sigma $ has the following analytic structure (see Figure (\ref{fig:sigmafigures})).
\begin{enumerate}
\item When $ | \Delta | \leq 1 $, $ \sigma(s,t) $ is strictly convex and smooth for all $ -1/2 < s, t < 1/2 $, with corner singularities at the boundary.
\item When $ \Delta < -1 $, $ \sigma $ is convex with corner singularities near the boundary, and an additional conical singularity at $ (s,t)=(0,0)$. It is smooth and strictly convex away from $(0,0)$ and the boundary.

\item When $ \Delta > 1 $, $ \sigma $ is convex with corner singularities at the boundary and in addition corner singularity along $ s = t, -t_0<t<t_0 $ when $ a > b $, and along $ s = -t, -t_0<t<t_0 $ when $ a < b $,
    for certain $t_0$ (which is a function of weights). The function $\sigma$ is smooth and strictly convex away from these singularities.
\end{enumerate}
\begin{figure}[h]
\includegraphics[scale=.25]{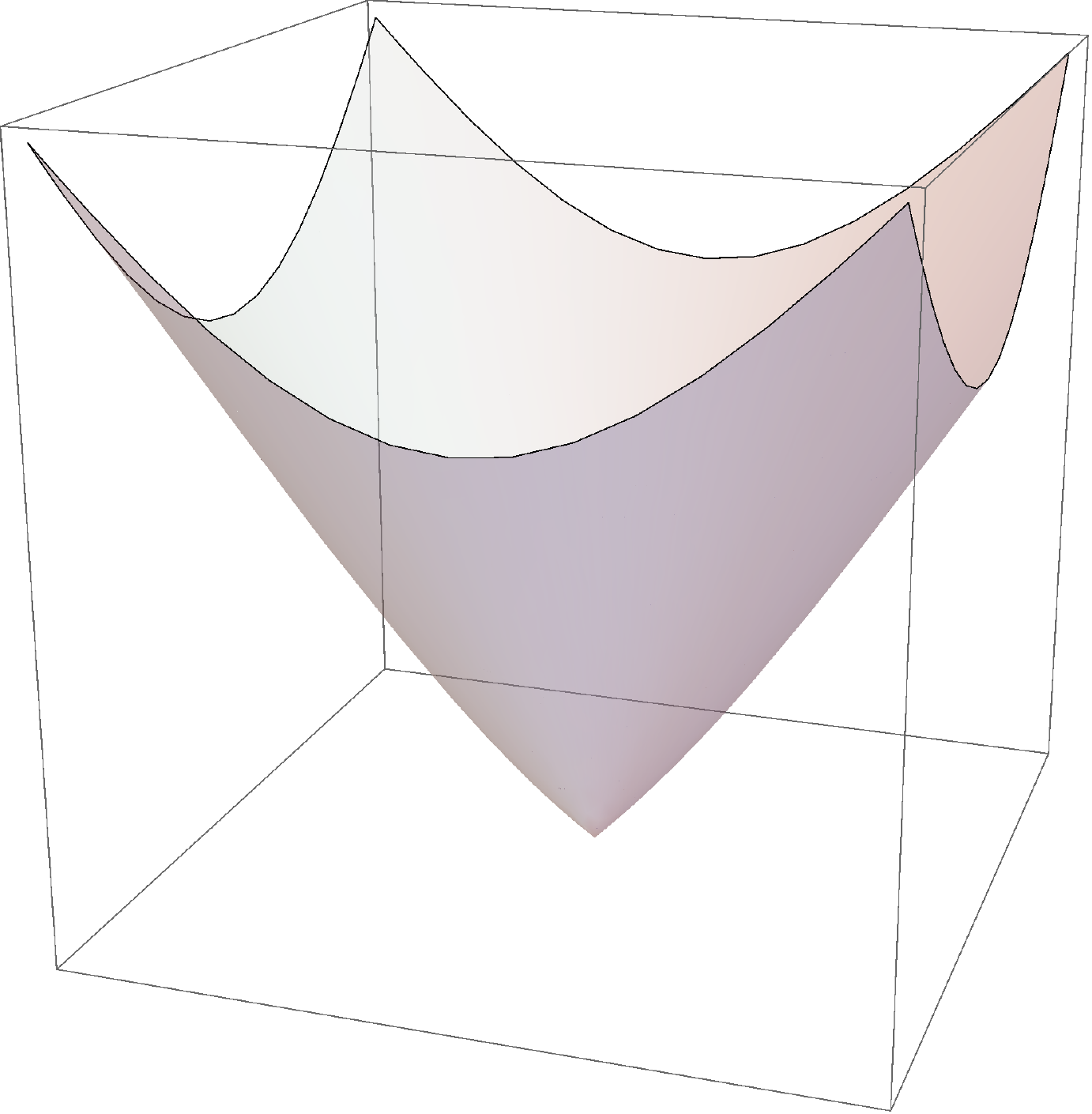} \;
\includegraphics[scale=.25]{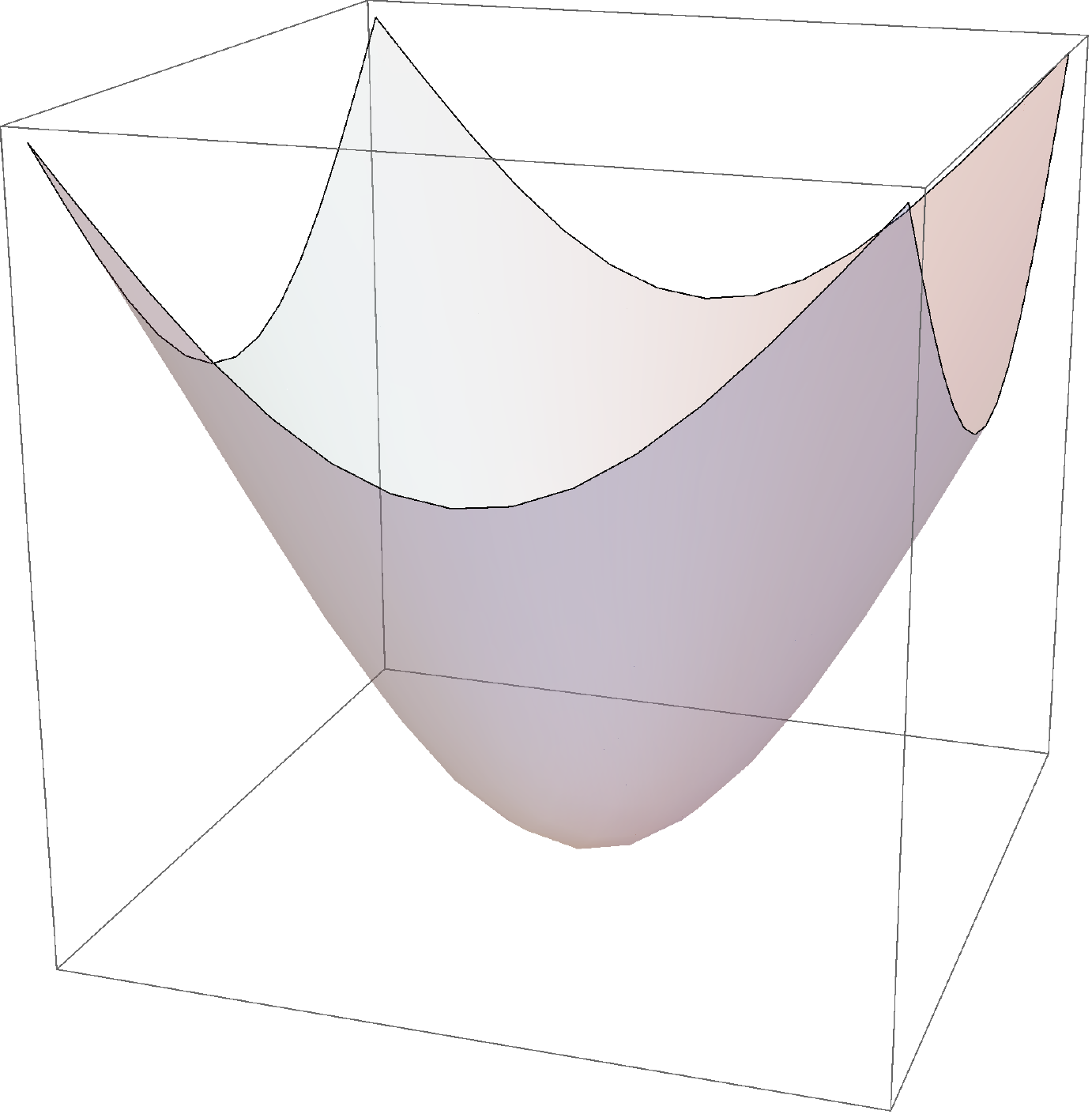} \;
\includegraphics[scale=.25]{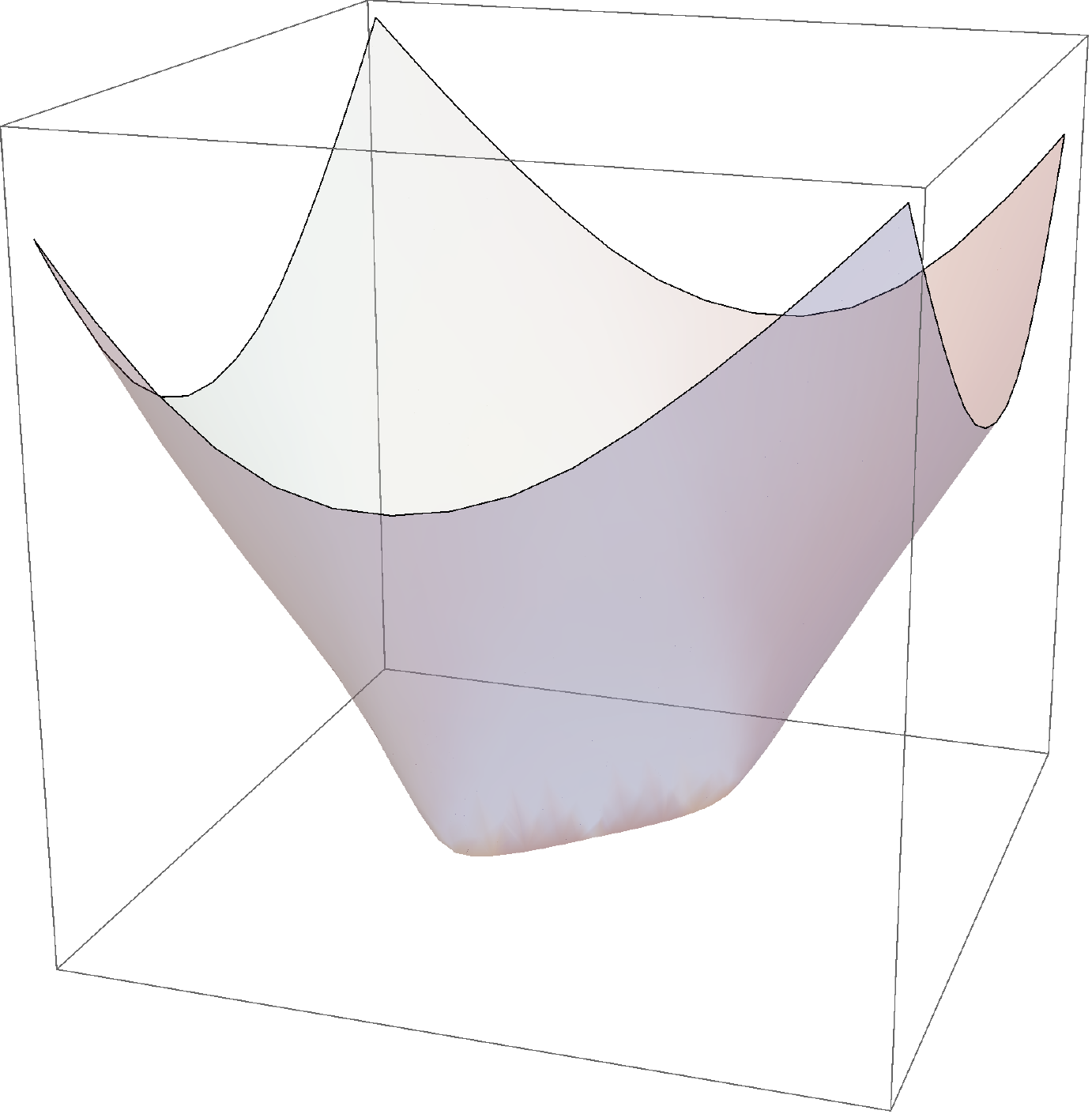}
\caption{Sketch of $ \sigma $ for $ \Delta < -1 $, $ | \Delta | < 1 $, and $ \Delta > 1 $.}\label{fig:sigmafigures}
\end{figure}
An important fact regarding the convexity of $ \sigma $ is the following proposition based on \cite{NK} :
\begin{proposition}\label{prop:hesssigma}
The Hessian $ \det \left( \partial_i \partial_j \sigma \right) $ only depends on $s, t$ and $ \Delta $.
It does not depend on spectral parameter.
\end{proposition}

\noindent 
The proof is not difficult to derive from \cite{NK}, but is crucial in the proof of commutativity of Hamiltonians in Section \ref{Ham-6v-section}.

\subsection{The Thermodynamic Limit and Limit Shapes}

Now let $ M $ be either a planar domain or a cylinder, and let $ \{ M^{(n)}\}_{(n = 1 \cdots \infty)} $ be a sequence of embedded surface graphs as described above. Let $ \{ \chi^{(n)} \} $ be a sequence of normalized boundary height functions on $ \{ M^{(n)} \} $.
The sequence $\{ \chi^{(n)} \} $ is said to be stabilizing if it converges in the uniform metric to a function $ \chi: \partial M \rightarrow \mathbb{R} $. As in the case of the torus one can take for granted that the free energy exists
with stabilizing boundary conditions:
\begin{align*}
f_{M,\chi} = \lim_{n\rightarrow \infty} \epsilon_n^2 \log(Z_{M^{(n)}, \chi^{(n)}}) ,
\end{align*}
where the limit does not depend on the choice of the sequence $\{\epsilon_n\}$.

Similarly to \cite{CKP}, one can argue that the limit of normalized average height function exists:
\begin{align*}
h_0(x,y) =  \lim_{n \rightarrow \infty} \langle h^{(n)}(x,y) \rangle
\end{align*}
Here $ \langle A \rangle $ is the expectation value of a functional on the space of height functions
with respect to the 6-vertex Boltzmann distribution (assuming the identification of
height functions and states of the 6-vertex model).
Moreover, for any  $(x_1, y_1), \dots, (x_k,y_k)$:
\[
\lim_{n\to \infty} \langle h^{(n)}(x_1,y_1), \dots, h^{(n)}(x_k,y_k) \rangle = h_0(x_1,y_1), \dots,  h_0(x_k,y_k)
\]

By techniques essentially identical to \cite{CKP} where height functions were studied for dimer models, it can be shown that the limiting free energy and the limiting height function of the six vertex model in the thermodynamic limit are determined by the following variational problem.

Let us consider first the planar domain. Let $ \mathfrak{H}(\chi) $ be the space of once differentiable functions
on a connected simply connected domain $D\subset \RR^2$ with slopes $ \partial_x h, \partial_y h \in [-1/2,1/2] $, and with fixed tangential derivative at the boundary $\pa_\tau h=\pa_\tau \chi$.  Then the limiting free energy in the thermodynamic limit $ n \rightarrow \infty $ on the planar domain is:
\begin{align} \label{eq:vareq}
f_{D, \chi} = \max_{h \in \mathfrak{H}(\chi) } \int_{D} \sigma(\partial_x h, \partial_ y h) \; dx \; dy
\end{align}
The maximizer is the limit shape $h_0(x,y)$ , i.e. the limit of the
expectation value of the normalized height function.

\begin{remark}
In the presence of magnetic fields the extra term $H\int\int_D \pa_xh \; dx dy +V\int\int_D \pa_xh \; dx dy $
has to be added to the action function functional
\[
S[h]=\int_{D} \sigma(\partial_x h, \partial_ y h) \; dx \; dy
\]
However, since this extra term depends only on boundary values of $h$, it
does not change the minimizer of $S[h]$.
\end{remark}

Turning next to the cylinder, note first that the dependence of the partition function for a cylinder on the horizontal magnetic field is very simple, see (\ref{cyl-mag-field}), so we will assume $H=0$. We assume the boundary condition satisfies $\pa_\tau \chi_i(y+L)= \pa_\tau \chi_i(y)$. Define $ \mathfrak{H}(\chi_1,\chi_2) $ to be the space of once differentiable
functions on  $ C_{T,L}=[0T]\times [0,L]$  (the cylinder with the branch cut along $y=0$)
such that $\pa_{x,y}h(x,y)$ are periodic, $\pa_{x,y}h(x,y+L)=\pa_{x,y}h(x,y)$, and $\pa_yh(0,y)=\pa_y \chi_1(y)$, $\pa_yh(T,y)=\pa_y \chi_2(y)$ .

\begin{proposition}\label{prop:limshape}The free energy for the cylinder is
\begin{align} \label{eq:vareq2}
f_{C_{T,L}, \chi_1, \chi_2} =  \int_{C_{T, L}} \sigma(\partial_x h_0, \partial_ y h_0) \; dx \; dy+V \int_{C_{T, L}} \pa_x h_0 \; dx \; dy
\end{align}
where $h_0(x,y)$ is the unique maximizer of the functional
\begin{align} \label{eq:action}
S[h] = \int_0^T  \int_0^L \left(\sigma( \partial_x h, \partial_y h)   +V\pa_xh \right)\; dx\; dy
\end{align}
on $\mathfrak{H}(\chi_1,\chi_2)$.
\end{proposition}

\begin{remark} If we want to instead impose Dirichlet boundary conditions on both ends of
the cylinder, we would impose $h(0,y)= \chi_1(y)$ and $h(T,y)= \chi_2(y)+C$ and would
vary over the height functions inside the domain and over $C \in \mathbb{R} $.
\end{remark}

On the regions where the height function is smooth, $h_0$ satisfies the Euler-Lagrange equation:
\begin{align} \label{eq:heightEL}
\partial_1^2 \sigma( \partial_x h, \partial_y h ) \; \partial _x^2 h + \partial_1 \partial_2 \sigma( \partial_x h, \partial_y h ) \; \partial _x \partial_y h + \partial_2^2 \sigma( \partial_x h, \partial_y h ) \; \partial _y^2 h = 0
\end{align}
By the strict convexity of $ \sigma $ on smooth regions, this is an elliptic partial differential equation.

In the context of height functions we will call the function $\sigma(s,t)$ the {\it surface tension function}.
It is a surface tension function if we would  study the gradient evolution of the height function
with energy functional $S[h]$.

\section{The Hamiltonian Framework and Commuting Hamiltonians} \label{Ham-6v-section}
In this section, we reformulate the variational principle (\ref{prop:limshape}) and second order PDE (\ref{eq:heightEL}) in the Hamiltonian (first order) framework. The Hamiltonian function in this framework should be regarded as the semi-classical limit of the transfer matrix. We prove the main result that when the transfer matrices commute, the corresponding Hamiltonians Poisson commute.

\subsection{Hamiltonian Framework on the Cylinder} We focus on the variational principle on the cylinder $ C_{TL} $. We will think of the $ x$-direction on the cylinder as time. Let $\mathfrak{H}$ be the space of differentiable functions on $[0,L]$ with periodic derivatives.
As usual in Lagrangian mechanics the action $S$ given by (\ref{eq:action}) is the functional on paths $[0,T]\mapsto \mathfrak{H}$ such that $\pa_y h(0,y)=\chi_1(y)$ and $\pa_y h(T,y)=\chi_2(y)$, where $\chi_i$ are $L$-periodic.
The function $ \sigma $ is the Lagrangian density functional on $T\mathfrak{H}$ with $\pa_x h(x,y)$ being tangent vector at $h$. The free energy $ f_{C_{TL}, \eta_1, \eta_2 } $ is the Hamilton-Jacobi action determined by the action functional (\ref{eq:action}).

The Hamiltonian framework is defined in the usual way by passing from Lagrangian densities
on $T\mathfrak{H}$ to Hamiltonian densities on $T^*\mathfrak{H}$ via the Legendre transform. Consider $\{h(y)\}_{y=0}^L$ as a collection of "coordinate functions" on $T^*\mathfrak{H}$.
Let $\{p(y)\}_{y=0}^L$ be the collection of "coordinate functions" on $T^*_h \mathfrak{H}$.
Since tangent space $T_h\mathfrak{H}$ consists of periodic functions, we assume $p(y+L)=p(y)$.
The canonical Poisson bracket on functionals on $T^*_h \mathfrak{H}$ is defined by its values on the coordinate functions $ p $ and $ h $ is:
\begin{align*}
\{h(y), p(y') \} &= \delta(y - y')  \\
\{ p(y), p(y') \} &= 0  \\
\{ h(y), h(y') \} &= 0
\end{align*}

The Hamiltonian density $ \tau $ is the Legendre transform with respect to the first argument of the surface tension:
\begin{align}
\tau_u(p, \xi) = \max_{\nu \in \mathbb{R} } \big( p \; \nu - \sigma_u(\nu, \xi)+V\nu \big)
\end{align}
The Hamiltonian describing the evolution of the height function from
the left end of the cylinder to the right end is given by:
\begin{align}\label{Ham}
H_u(p, h) = \int_{0}^L \tau_u\big(p(y)+V, \partial_y h(y) \big) \; dy
\end{align}
for $ p, h: [0,L] \rightarrow \mathbb{R} $. Here and in the previous formula
we emphasized the dependence of $\sigma$ and $\tau$ on the spectral parameter in Baxter's
parametrization since, as we will see below, this is the parameter of commutative family
of Hamiltonians. Since the dependance of the Hamiltonian on $V$ appears as a change of variable (symplectomorphism)
$p\to p+V, \pa_yh \to \pa_yh$ we will assume from now on that $V=0$.

The Hamiltonian equations of motion are:
\begin{equation} \label{eq:HamEq}
\begin{aligned}
\partial_x p(x,y) &= \{ p(x,y), H \} =  - \frac{\delta H}{ \delta h(y) } \big( p(x,y), h(x, y) \big) \\
\partial_x h(x,y) &= \{ h(x,y, H \} =  \frac{\delta H}{ \delta p(y) } \big( p(x,y), h(x, y) \big)
\end{aligned}
\end{equation}
It is easy to check that they are equivalent to the Euler-Lagrange equations for
the height function (\ref{eq:heightEL}).

Note that in the Hamiltonian framework equations of motion are
extema of the action functional:
\begin{align}
S[p, h] = \int_0^T \int_0^L p(x,y) \partial_y h(x,y) - \tau_u \big( p(x,y), \partial_y h(x,y) \big) \; dy \; dx
\end{align}
for $ p, h: C_{T,L} \rightarrow \mathbb{R} $.

\subsection{Commuting Hamiltonians}In this section we will prove that
Hamiltonians (\ref{Ham}) form a commutative family, i.e. for any $u$ and $v$,
\[
\{ H_u, H_v \} = 0
\]
We first compute the Poisson bracket of two $ H_u $ and $ H_v $ with density $ \tau_u $ and $ \tau_v $:
\begin{lemma}
\begin{equation}
\begin{aligned} \label{eq:hambrack}
\{ H_u, H_v \} =&  \int_0^L \big( \partial_1^2 \tau_u(y) \; \partial_2 \tau_v(y) - \partial_1^2 \tau_v(y) \; \partial_2 \tau_u(y) \big)\;  \partial_y p(y)  \\ & \; \; \; + \big( \partial_1 \partial_2 \tau_u(y) \; \partial_2 \tau_v(y) - \partial_1 \partial_2 \tau_v(y) \; \partial_2 \tau_u(y) \big) \; \partial_y ^2 h(y) \; dy
\end{aligned}
\end{equation}
Here we abbreviated the arguments, that is, by $ \tau_u(y) $ we mean $ \tau_u(p(y), \partial_y h(y) ) $.
\end{lemma}

\begin{proof}
It is a straightforward computation:
\begin{align*}
\{ H_u, H_v \} = & \int_0^L \int_0^L \big\{ \tau_u(p(y), \partial_y h (y)), \tau_v(p(y'), \partial_{y'} h(y')) \big\} \; dy' \; dy\\\
=& \int_0^L\int_0^L \left(\partial_1 \tau_u(y) \; \partial_2 \tau_v(y') \; \{ p(y), \partial_{y'} h(y') \} + \partial_2 \tau_u(y) \; \partial_1 \tau_v(y') \; \{ \partial_y h(y), p(y') \} \right) \; dy' \; dy \\
=& \int_0^L\int_0^L  \left( -\partial_1 \tau_u(y) \; \partial_2 \tau_v(y') \;  \partial_{y} \delta(y -y')+ \partial_2 \tau_u(y)\; \partial_1 \tau_v(y') \; \partial_y \delta(y-y')  \right)\; dy' \; dy \\
=& \int_0^L\int_0^L  \left(\left( \partial_y \partial_1 \tau_u(y) \right)  \partial_2 \tau_v(y') \delta(y -y') -  \partial_2 \tau_u(y)\left( \partial_{y'} \partial_1 \tau_v(y')\right) \delta(y-y') \right) \; dy' \; dy \\
=& \int_0^L  \left(\left( \partial_1^2 \tau_u \partial_y p + \partial_1 \partial_2 \tau_u \partial_y^2 h \right)  \partial_2 \tau_v -  \partial_2 \tau_u \left( \partial_1^2 \tau_v \partial_y p + \partial_1 \partial_2 \tau_v \partial_y^2 h \right)  \right)\; dy \\
=&  \int_0^L  \left(\big( \partial_1^2 \tau_u \; \partial_2 \tau_v - \partial_1^2 \tau_v \; \partial_2 \tau_u \big)\;  \partial_y p  + \big( \partial_1 \partial_2 \tau_u \; \partial_2 \tau_v - \partial_1 \partial_2 \tau_v \; \partial_2 \tau_u \big) \; \partial_y ^2 h \right)\; dy
\end{align*}
\end{proof}

\begin{proposition} \label{prop:comham}
For any $u$ and $v$, the Hamiltonians $ H_u $ and $ H_v $ Poisson commute if  $\text{Hess}(\sigma_u) = \text{Hess}(\sigma_v)$.
\end{proposition}
\begin{proof}

From Lemma  \ref{eq:hambrack} we have:
\begin{align*}
\{ H_u, H_v \} &=  \int_0^L \left(A \left( p(y), \partial_y h(y)  \right) \partial_y p + B \left( p(y), \partial_y h(y) \right) \partial_y^2 h\right) \; dy
\end{align*}
where
\begin{align*}
A &=  \partial_1^2 \tau_u \partial_2 \tau_v - \partial_1^2 \tau_v \partial_2 \tau_u \\
B &=  \partial_1 \partial_2 \tau_u \partial_2 \tau_v - \partial_1 \partial_2 \tau_v \partial_2 \tau_u
\end{align*}
If
\begin{align*}
\partial_2 A(s_1, s_2) - \partial_1 B(s_1, s_2) = 0
\end{align*}
on $D=\{(s_1,s_2)| -1/2\leq s_i\leq 1/2\}$, then there exists $ \mathcal{F}(s_1,s_2 ) $ such that:
\begin{equation*}
\partial_1 \mathcal{F} = A, \ \
\partial_2 \mathcal{F} = B
\end{equation*}
and so:
\begin{align*}
\{ H_u, H_v \} &= \int_0^L \left(\partial_1 \mathcal{F}( p(y), \partial_y h(y) ) \; \partial_y p +  \partial_2 \mathcal{F}( p(y), \partial_y h(y) ) \; \partial_y^2 h\right) \; dy \\
&= \int_0^L \partial_y \mathcal{F}( p(y), \partial_y(h) ) \; dy = 0
\end{align*}

Now we compute $ \partial_2 A - \partial_1 B $:
\begin{align*}
\partial_2 A - \partial_1 B =& \;  \partial_2 \left( \partial_1^2 \tau_u \partial_2 \tau_v - \partial_1^2 \tau_v \partial_2 \tau_u \right) - \partial_1  \left(  \partial_1 \partial_2 \tau_u \partial_2 \tau_v - \partial_1 \partial_2 \tau_v \partial_2 \right) \\
=& \; (\partial_1^2 \partial_2 \tau_u \partial_2 \tau_v + \partial_1^2 \tau_u \partial_2^2 \tau_v - \partial_1^2 \partial_2 \tau_v \partial_2 \tau_u - \partial_1^2 \tau_v \partial_2^2 \tau_u  )  \\ & -  ( \partial_1^2 \partial_2 \tau_u \partial_2 \tau_v + \partial_1 \partial_2 \tau_u \partial_1 \partial_2 \tau_v - \partial_1^2 \partial_2 \tau_v \partial_2 \tau_u  - \partial_1 \partial_2 \tau_v \partial_1 \partial_2 \tau_u ) \\
=& \partial_1^2 \tau_u \partial_2^2 \tau_v - \partial_1^2 \tau_v \partial_2^2 \tau_u \\
=& \partial_1^2 \tau_u \; \partial_1^2 \tau_v \; \left( \frac{ \partial_2^2 \tau_v }{\partial_1^2 \tau_v } - \frac{ \partial_2^2 \tau_u }{\partial_1^2 \tau_u } \right) \\
=& \partial_1^2 \tau_u \; \partial_1^2 \tau_v  \left( \text{Hess}(\sigma_v) - \text{Hess}(\sigma_u) \right)
\end{align*}
In the last line, we used Lemma \ref{Hess-Lemma} from the appendix.

\end{proof}
\noindent Together, propositions \ref{prop:hesssigma} and  \ref{prop:comham} imply:

\begin{corollary}
The Hamiltonians $ H_u $ of the 6-vertex model form a Poisson commuting family.
\end{corollary}

\begin{remark} The existence of Hamiltonians $H_u$ is due to limit shape phenomenon,
which can be regarded as a semiclassical nature of the thermodynamic limit. The Poisson
commutativity of Hamiltonians is a consequence of commutativity of transfer-matrices and
thus, of the Yang-Baxter equation for the weights of the 6-vertex model. One can show that
for any model with limit shape phenomenon and commutative family of transfer-matrices
we have a family of Poisson commutative Hamiltonians for limit shapes on a cylinder.
The details will be given elsewhere.
\end{remark}

\section{The Dimer Model, Burger's Equation.}\label{dimer-section}

\subsection{Dimer models}

Here we will recall some basic facts on dimer models on surface bipartite graphs.
More details can be found in \cite{KOS}\cite{KO}\cite{CR}.

A dimer configuration $\mathcal{D} $ on a graph $\Gamma$ is a  subset of edges
called dimers, such that each vertex of valency greater then $1$ is adjacent to exactly one dimer.
One valent vertices need not belong to a dimer, and we will refer the set of $1$-valent
vertices as a boundary $\pa \Gamma$ of $\Gamma$. In other words, a dimer configuration is a perfect matching on the set of vertices connected by edges. A dimer configuration is illustrated in Figure \ref{dimers}.

\begin{figure}[h]
\begin{tikzpicture}[scale=1]
\coordinate (d0) at (-1,.1);
\coordinate (d1) at (-.4,.7);
\coordinate (d2) at (.6,.6);
\coordinate (d3) at (1.1,-0);
\coordinate (d4) at (.5,-.7);
\coordinate (d5) at (-.4,-.65);

\coordinate (b0) at (-2,-.2);
\coordinate (b1) at (-.4,1.4);
\coordinate (b2) at (1.2,1);
\coordinate (b3) at (2,0);
\coordinate (b4) at (.6,-1.4);
\coordinate (b5) at (-.6,-1.3);

\coordinate (b01) at (-1.3,.6);
\coordinate (b01a) at (-2,.2);
\coordinate (b01b) at ( -1,1.3);
\coordinate (b12) at (.5,1.3);
\coordinate (b23) at (1.7,.6);
\coordinate (b34) at (1.6,-.8);
\coordinate (b45) at (-.1,-1.4);
\coordinate (b50) at (-1.5,-1.0);

\draw [dashed] plot [smooth cycle] coordinates {(b0)(b01a)(b01)(b01b)(b1)(b12)(b2)(b23)(b3)(b34)(b4)(b45)(b5)(b50)};
\fill [color= black!5!white] plot [smooth cycle] coordinates {(b0)(b01a)(b01)(b01b)(b1)(b12)(b2)(b23)(b3)(b34)(b4)(b45)(b5)(b50)};

\draw (d0)--(d1);
\draw (d1)--(d2);
\draw [ultra thick] (d2)--(d3);
\draw (d3)--(d4);
\draw (d4)--(d5);
\draw [ultra thick] (d5)--(d0);

\draw (d0)--(b0);
\draw [ultra thick] (d1)--(b1);
\draw (d2)--(b2);
\draw (d3)--(b3);
\draw [ultra thick] (d4)--(b4);
\draw (d5)--(b5);

\coordinate (s) at (.05,.05);

\fill (d0) circle (2pt);
\fill[white] (d1) circle (2pt);\draw (d1) circle (2pt);
\fill (d2) circle (2pt);
\fill[white] (d3) circle (2pt);\draw (d3) circle (2pt);
\fill (d4) circle (2pt);
\fill[white] (d5) circle (2pt);\draw (d5) circle (2pt);

\fill ($ (b0) - (s) $ ) rectangle  ($ (b0) + (s) $ );
\fill ($ (b1) - (s) $ ) rectangle  ($ (b1) + (s) $ );
\fill ($ (b2) - (s) $ ) rectangle  ($ (b2) + (s) $ );
\fill ($ (b3) - (s) $ ) rectangle  ($ (b3) + (s) $ );
\fill ($ (b4) - (s) $ ) rectangle  ($ (b4) + (s) $ );
\fill ($ (b5) - (s) $ ) rectangle  ($ (b5) + (s) $ );

\end{tikzpicture}

\caption{A dimer configuration on a surface graph.}
\end{figure}
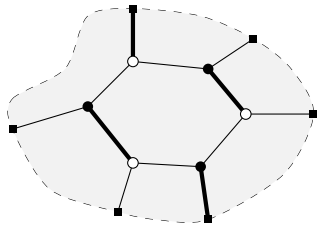 \label{dimers}

To define a dimer model on $\Gamma$, one should fix a weight function, i.e.
a mapping from edges of the graph to non-negatives real number:
\[
W: e\mapsto w(e)
\]
Then the weight of a dimer configuration $\mathcal{D}$ is defined
as
\[
W(\mathcal{D})=\prod_{e\in \mathcal{D}} w(e)
\]
The dimer configuration on edges adjacent to $1$-valent vertices
is called the boundary value $\pa \mathcal{D}$ of $\mathcal{D}$.

The partition function with boundary value $\mathcal{B}$ is
\[
Z_{\Gamma, \mathcal{B}}=\sum_{\mathcal{D}: \pa \mathcal{D}=\mathcal{B}} W(\mathcal{D})
\]

Here we will consider only dimer models on bipartite surface graphs. Recall that such a graph is a triple $(\Gamma, \phi, \Sigma)$ where $\Gamma$ is
a bipartite graph with vertices of possible valence $1, 3, 4,\dots$,
$\Sigma$ is a compact oriented surface, $\phi: \Gamma \to \Sigma$ is an
embedding such that all 1-valent vertices are mapped to $\pa \Sigma$.
So far we do not require any other structure on $\Sigma$ such as metric.

\subsection{Height Functions} \label{dimer-height}

 Two dimer configurations $\mathcal{D}$ and $\mathcal{D}_0$
on a bipartite surface graph $\Gamma$ define
a system of oriented composition cycles with edges occupied by $\mathcal{D}$ oriented from black
vertices to white vertices and with edges occupied by $\mathcal{D}_0$ from white to black.
This is illustrated on Figure \ref{fig:dimers2}.
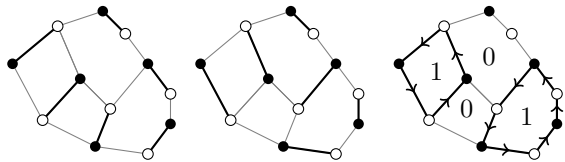
\begin{figure}[h]
\begin{tikzpicture}[scale=1]
\coordinate (b0) at (-.8,.1);
\coordinate (b1) at (.1,-.1);
\coordinate (b2) at (.4,.8);
\coordinate (b3) at (.3,-1);
\coordinate (b4) at (1,.1);
\coordinate (b5) at (1.3,-.7);
\coordinate (w0) at (-.2,.6);
\coordinate (w1) at (-.4,-.65); 
\coordinate (w2) at (.5,-.5);
\coordinate (w3) at (.7, .5);
\coordinate (w4) at (1,-1.1);
\coordinate (w5) at (1.3,-.3);

\draw[gray] (b0)--(w1)--(b1)--(w0)--(b0);
\draw[gray] (w1)--(b1)--(w2)--(b3)--(w1);
\draw[gray] (w0)--(b1)--(w2)--(b4)--(w3)--(b2)--(w0);
\draw[gray] (b3)--(w4)--(b5)--(w5)--(b4);

\draw[thick](b0)--(w0);
\draw[thick](b1)--(w1);
\draw[thick](b3)--(w2);
\draw[thick](b4)--(w5);
\draw[thick](b5)--(w4);
\draw[thick](b2)--(w3);

\foreach \i in {0,...,5}
{
\fill (b\i) circle (2pt);
\fill[white] (w\i) circle (2pt);\draw (w\i) circle (2pt);
}
\end{tikzpicture}\;
\begin{tikzpicture}[scale=1]
\coordinate (b0) at (-.8,.1);
\coordinate (b1) at (.1,-.1);
\coordinate (b2) at (.4,.8);
\coordinate (b3) at (.3,-1);
\coordinate (b4) at (1,.1);
\coordinate (b5) at (1.3,-.7);
\coordinate (w0) at (-.2,.6);
\coordinate (w1) at (-.4,-.65); 
\coordinate (w2) at (.5,-.5);
\coordinate (w3) at (.7, .5);
\coordinate (w4) at (1,-1.1);
\coordinate (w5) at (1.3,-.3);

\draw[gray] (b0)--(w1)--(b1)--(w0)--(b0);
\draw[gray] (w1)--(b1)--(w2)--(b3)--(w1);
\draw[gray] (w0)--(b1)--(w2)--(b4)--(w3)--(b2)--(w0);
\draw[gray] (b3)--(w4)--(b5)--(w5)--(b4);

\draw[thick](b0)--(w1);
\draw[thick](b1)--(w0);
\draw[thick](b3)--(w4);
\draw[thick](b5)--(w5);
\draw[thick](b4)--(w2);
\draw[thick](b2)--(w3);

\foreach \i in {0,...,5}
{
\fill (b\i) circle (2pt);
\fill[white] (w\i) circle (2pt);\draw (w\i) circle (2pt);
}
\end{tikzpicture}
\;
\begin{tikzpicture}[scale=1]
\coordinate (b0) at (-.8,.1);
\coordinate (b1) at (.1,-.1);
\coordinate (b2) at (.4,.8);
\coordinate (b3) at (.3,-1);
\coordinate (b4) at (1,.1);
\coordinate (b5) at (1.3,-.7);
\coordinate (w0) at (-.2,.6);
\coordinate (w1) at (-.4,-.65); 
\coordinate (w2) at (.5,-.5);
\coordinate (w3) at (.7, .5);
\coordinate (w4) at (1,-1.1);
\coordinate (w5) at (1.3,-.3);

\draw[gray] (b0)--(w1)--(b1)--(w0)--(b0);
\draw[gray] (w1)--(b1)--(w2)--(b3)--(w1);
\draw[gray] (w0)--(b1)--(w2)--(b4)--(w3)--(b2)--(w0);
\draw[gray] (b3)--(w4)--(b5)--(w5)--(b4);

\draw[thick, decoration={markings,mark=at position .6 with {\arrow[scale=1]{<}}},
    postaction={decorate}](b0)--(w0);
\draw[thick, decoration={markings,mark=at position .6 with {\arrow[scale=1]{<}}},
    postaction={decorate}](b1)--(w1);
\draw[thick, decoration={markings,mark=at position .6 with {\arrow[scale=1]{<}}},
    postaction={decorate}](b3)--(w2);
\draw[thick, decoration={markings,mark=at position .6 with {\arrow[scale=1]{<}}},
    postaction={decorate}](b4)--(w5);
\draw[thick, decoration={markings,mark=at position .6 with {\arrow[scale=1]{<}}},
    postaction={decorate}](b5)--(w4);
\draw[thick, decoration={markings,mark=at position .6 with {\arrow[scale=1]{<}}},
    postaction={decorate}](w1)--(b0);
\draw[thick, decoration={markings,mark=at position .6 with {\arrow[scale=1]{<}}},
    postaction={decorate}](w0)--(b1);
\draw[thick, decoration={markings,mark=at position .6 with {\arrow[scale=1]{<}}},
    postaction={decorate}](w4)--(b3);
\draw[thick, decoration={markings,mark=at position .6 with {\arrow[scale=1]{<}}},
    postaction={decorate}](w5)--(b5);
\draw[thick, decoration={markings,mark=at position .6 with {\arrow[scale=1]{<}}},
    postaction={decorate}](w2)--(b4);

\node at (-.3,0){\footnotesize{$1$}};
\node at (.9,-.6){\footnotesize{$1$}};
\node at (.4,.2){\footnotesize{$0$}};
\node at (.1,-.5){\footnotesize{$0$}};

\foreach \i in {0,...,5}
{
\fill (b\i) circle (2pt);
\fill[white] (w\i) circle (2pt);\draw (w\i) circle (2pt);
}

\
\end{tikzpicture}
\caption{Dimer configurations $ \mathcal{D}_0$ and $ \mathcal{D} $. On the right side is the composition cycle $ ( \mathcal{D}, \mathcal{D}_0) $ and $ \theta_{\mathcal{D}, \mathcal{D}_0} $ } \label{fig:dimers2}
\end{figure} 

For a planar surface graph on a connected simply connected domain, the
oriented composition cycle $(\mathcal{D}, \mathcal{D}_0)$ defines a function $\theta_{\mathcal{D}, \mathcal{D}_0}$
on faces of $\Gamma\subset D$
with the composition cycles being its level curves; the height function increments when crossing a cycle oriented from left to right. This function is defined up to a constant,
which can be fixed by fixing its value at a chosen reference face.

As in the $6$-vertex model, for surface graphs on non simply-connected surfaces, the height function
does not exist as a function of faces. We will focus on the case of dimer models on a cylinder and torus with branch cuts, where the height function can be defined globally. 

From here on, we assume a fixed reference dimer configuration $ \mathcal{D}_0 $, so that there is a bijection of dimer configurations $ \mathcal{D} $ and height functions $ \theta_\mathcal{D} $. 

For  any path $C$ (on a dual graph to $\Gamma\subset \Sigma$) connecting two faces,
the difference $\Delta_C\theta$ between the value of $\theta$ at the end points of $C$
does not depend on $C$ (we assume that $C$ does not cross branch cuts). If $C$ is a closed contour on $\Sigma$,
the change of $\theta$ around $C$ depends only on the homology class of $C$.
For non simply-connected surfaces this allows for an addition weight $z_C$ for each basis cycle $ C$ of  $ \Sigma $.

It is not difficult to see that the dimer partition function for a dimer model on a bipartite surface graph
with extra weights corresponding to the collection of cycles $\{C_i\}$ can be written in terms of
height functions as 
\begin{equation}\label{T}
Z_\Gamma(z_1, \cdots, z_n)= \frac{1}{W(\mathcal{D}_0) } \sum_\theta W(\theta) \prod_i z_i^{\Delta_{C_i} \theta }
\end{equation}
where $z_i$ is the weight for the cycle $C_i$. The sum is taken over all possible height functions on
the graph and
\begin{equation}\label{hf-weights}
W(\theta)=\prod_f q_f^{ \theta(f) }, \ \  q_f=\prod_{e\subset \pa f} w(e)^{\varepsilon(e,f)}
\end{equation}
Here the product is taken over all 2-cells of $\Gamma\subset \Sigma$. Edges of $\Gamma$ are oriented from black
vertices to white and $\varepsilon(e,f)$ is $1$ when the orientation of $e$ agrees with the orientation of
$\pa f$ induced by the orientation of $\Sigma$,  and $-1$ otherwise. 

For a torus with extra weights $z$ and $w$ corresponding
to vertical and horizontal cycles we have
\begin{equation}\label{T}
Z_\Gamma(z,w)= \frac{1}{W(\mathcal{D}_0) } \sum_{\theta} W(\theta)  z^{\Delta_b\theta}w^{\Delta_a \theta}
\end{equation}
where $a$ is a "vertical" cycle and $b$ is a "horizontal" one. 

For a cylinder we have
\begin{equation}\label{C}
Z_{\Gamma, \theta_1, \theta_2}(z)=\sum_{\theta} W(\theta) w^{\Delta_c\theta}
\end{equation}
where $c$ is a "vertical" cycle, $\theta_1$ and $\theta_2$ are boundary tangential differences of $\theta$ on two
ends of the cylinder, and the sum is taken over all height functions with these boundary conditions.
See Appendix \ref{h-f} for more details height functions.

As in the $6$-vertex model we will refer to $H=\log z$ and $V=\log w$ as horizontal and vertical magnetic fields.

\subsection{Thermodynamic limit}As for the 6-vertex model we will consider only three types
of surfaces: a place,  cylinder and a torus. To define our planar graphs
fix a $ \mathbb{Z}^2 $ periodic lattice ${\mathcal L}$ with the fundamental domain $\Gamma_0$.
Locally our plane graphs have the structure of ${\mathcal L}$. To describe
our sequences of graphs globally, as for the 6-vertex model, fix an embedding
$\phi_\epsilon: {\mathcal L}\to \RR^2$ with mesh $\epsilon$. The definition of the
mesh may vary, for example for a square lattice, as in the 6-vertex model, we
can require that $\epsilon$ is the Euclidean length of images of generators.

Here we will focus on the free energy for a cylinder. For simplicity let us assume
that translation axes of $\phi_\epsilon(\mathcal{L})$ are parallel to $x$- and $y$- axes.
In other words $\phi_\epsilon(\mathcal{L})$ looks like a plane graph from Figure \ref{fig:PeriodicDimerGraph}.
After $\phi_\epsilon$ is fixed we define $C^\epsilon_{T,L}=[C_{T,L}\cap \phi_\epsilon({\mathcal L})]$.
The graph is closed in the vertical direction and its boundary $1$-valent vertices belong to
two boundaries of the cylinder.

\begin{figure}[h]
\begin{tikzpicture}[scale=1]

\coordinate (t0) at (0,0);
\draw ($ (0,0) + (t0)  $) rectangle ($ (1,1) + (t0) $);
\node at ( $ (.5,.5) + (t0) $ ) {$ \Gamma_0 $};
\draw  ( $ (1,.65) + (t0) $ ) --($ (1.2,.65)  + (t0) $);
\draw ($(1,.35)+ (t0) $ ) -- ($(1.2,.35)+ (t0) $ );
\draw ($(0,.65)+ (t0) $ ) -- ($(-.2,.65)+ (t0) $ );
\draw ($(0,.35)+ (t0) $ ) --($ (-.2,.35)+ (t0) $ );
\draw ($(.65,1)+ (t0) $ ) -- ($(.65,1.2)+ (t0) $ );
\draw ($(.35,1) + (t0) $ )-- ($(.35,1.2)+ (t0) $ );
\draw ($(.65,0) + (t0) $ ) -- ($(.65,-.2)+ (t0) $ );
\draw ($(.35,0) + (t0) $ ) -- ($(.35,-.2)+ (t0) $ );

\coordinate (t0) at (3,0);
\draw ($ (0,0) + (t0)  $) rectangle ($ (1,1) + (t0) $);
\node at ( $ (.5,.5) + (t0) $ ) {$ \Gamma_0 $};
\draw  ( $ (1,.65) + (t0) $ ) --($ (1.2,.65)  + (t0) $);
\draw ($(1,.35)+ (t0) $ ) -- ($(1.2,.35)+ (t0) $ );
\draw ($(0,.65)+ (t0) $ ) -- ($(-.2,.65)+ (t0) $ );
\draw ($(0,.35)+ (t0) $ ) --($ (-.2,.35)+ (t0) $ );
\draw ($(.65,1)+ (t0) $ ) -- ($(.65,1.2)+ (t0) $ );
\draw ($(.35,1) + (t0) $ )-- ($(.35,1.2)+ (t0) $ );
\draw ($(.65,0) + (t0) $ ) -- ($(.65,-.2)+ (t0) $ );
\draw ($(.35,0) + (t0) $ ) -- ($(.35,-.2)+ (t0) $ );

\coordinate (t0) at (0,3);
\draw ($ (0,0) + (t0)  $) rectangle ($ (1,1) + (t0) $);
\node at ( $ (.5,.5) + (t0) $ ) {$ \Gamma_0 $};
\draw  ( $ (1,.65) + (t0) $ ) --($ (1.2,.65)  + (t0) $);
\draw ($(1,.35)+ (t0) $ ) -- ($(1.2,.35)+ (t0) $ );
\draw ($(0,.65)+ (t0) $ ) -- ($(-.2,.65)+ (t0) $ );
\draw ($(0,.35)+ (t0) $ ) --($ (-.2,.35)+ (t0) $ );
\draw ($(.65,1)+ (t0) $ ) -- ($(.65,1.2)+ (t0) $ );
\draw ($(.35,1) + (t0) $ )-- ($(.35,1.2)+ (t0) $ );
\draw ($(.65,0) + (t0) $ ) -- ($(.65,-.2)+ (t0) $ );
\draw ($(.35,0) + (t0) $ ) -- ($(.35,-.2)+ (t0) $ );

\coordinate (t0) at (3,3);
\draw ($ (0,0) + (t0)  $) rectangle ($ (1,1) + (t0) $);
\node at ( $ (.5,.5) + (t0) $ ) {$ \Gamma_0 $};
\draw  ( $ (1,.65) + (t0) $ ) --($ (1.2,.65)  + (t0) $);
\draw ($(1,.35)+ (t0) $ ) -- ($(1.2,.35)+ (t0) $ );
\draw ($(0,.65)+ (t0) $ ) -- ($(-.2,.65)+ (t0) $ );
\draw ($(0,.35)+ (t0) $ ) --($ (-.2,.35)+ (t0) $ );
\draw ($(.65,1)+ (t0) $ ) -- ($(.65,1.2)+ (t0) $ );
\draw ($(.35,1) + (t0) $ )-- ($(.35,1.2)+ (t0) $ );
\draw ($(.65,0) + (t0) $ ) -- ($(.65,-.2)+ (t0) $ );
\draw ($(.35,0) + (t0) $ ) -- ($(.35,-.2)+ (t0) $ );

\draw (0,0) rectangle (4,4);
\draw (1,1) rectangle (3,3);

\node at (.5,2) { $ \cdots $};
\node at (2,.5) { $ \cdots $};
\node at (3.5,2) { $ \cdots $};
\node at (2,3.5) { $ \cdots $};
\node at (2,2) { $ \cdots $ };
\end{tikzpicture}
\caption{A periodic graph with fundamental domain $ \Gamma_0 $.  } \label{fig:PeriodicDimerGraph}
\end{figure}
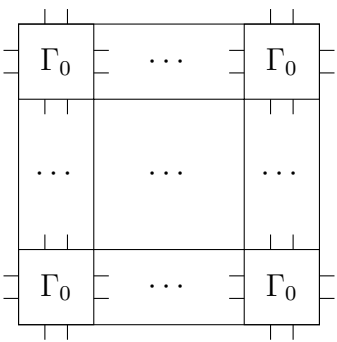 

We assume that weights of the dimer model are periodic with the fundamental domain $\Gamma_0$.
Define normalized height function $\Sigma$ as a piecewise constant on faces of $C^\epsilon_{T,L}\subset C_{T,L}$ with value
\[
h(x,y)=\epsilon \theta(f),
\]
if $(x,y)$ is in the interior of $f$.

For the periodic, periodically weighted graphs with fundamental domain $\Gamma_0$ as in Figure \ref{fig:PeriodicDimerGraph}
we assume that $\epsilon$ is the edge length for vertical and horizontal
edges connecting fundamental domains and the distances between neighboring horizontal and vertical edges
are also $\epsilon$. Similar assumptions will be held for a torus, which is the "closure" of a cylinder.

\subsubsection{The free energy for a torus}
From now on we assume that $\Gamma$ is a surface
periodic graph described above with the fundamental domain $\Gamma_0$ and with
periodic weights.

The density of the free energy exists for a torus is the limit of the dimer partition function (\ref{T}).
This is a corollary of the Kasteleyn formula for the dimer partition function \cite{Ka}.
The following answer was derived in \cite{KOS}:
\begin{align*}
f_{\mathbb{T}}(H,V) &=\lim_{\epsilon_n\to 0}\frac{1}{LT\epsilon_n^2}\log Z_{T^{(\epsilon_n)}_{M,N}}(e^H,e^V) \\ &= \frac{1}{( 2 \pi i )^2 } \int_{|z| = e^H} \int_{|w| = e^V} \log|P(z,w)| \; dz \; dw
\end{align*}
Here $P(z,w)$ is the characteristic polynomial of the dimer model on periodically
weighted graph
\[
P(z,w)=\sum_{\theta} W(\theta)z^{\Delta_b\theta}w^{\Delta_a\theta}(-1)^{ \Delta_a\theta\Delta_b\theta+\Delta_a\theta+ \Delta_b\theta}
\]
where the sum is taken over height functions on $\Gamma_0$ with toric boundary conditions.

The surface tension function $\sigma $ is the Legendre dual of the free energy:
\begin{align}\label{sigma-dim}
\sigma(s,t) = \max_{H,V} \big( s H + t V - f_{\mathbb{T}}(H,V) \big)
\end{align}
This is also the density of the conditional free energy for a torus with the average magnetization
constrained to $s$ in the horizontal direction and to $t$ in the vertical direction.

\subsubsection{The limit shape and Burgers' equation}
As for the 6-vertex model the free energy of a dimer model on a cylinder is determined by
the minimizer of
\[
S[h]= \int_0^T\int_0^L \sigma(\pa_xh, \pa_y h) \; dx \; dy
\]
where $\sigma(s,t)$ is the surface tension function (\ref{sigma-dim}).
For the same reason as in the 6-vertex case we assume that magnetic fields are zero.

The Euler-Lagrange equations for the limit shape are then:
\begin{equation} \label{eq:dimel}
\pa^2_1\sigma( \pa_x h, \pa_y h) \;  \pa^2_xh + 2 \; \pa_1\pa_2\sigma(\pa_x h, \pa_y h) \; \pa_x\pa_y h + \pa^2_2\sigma(\pa_x h, \pa_y h) \; \pa^2_yh = 0
\end{equation}

Define the new variables:
\begin{equation}
\begin{aligned} \label{eq:defzw}
z &= \mp \exp( i \pi \pa_y h + \pa_1\sigma(\pa_xh,\pa_yh) ) \\
w &= \pm \exp(- i \pi \pa_x h + \pa_2\sigma(\pa_xh,\pa_yh) )
\end{aligned}
\end{equation}
Then the Euler-Lagrange equations in terms of $ z $ and $ w $ are
\begin{equation} \label{eq:elzw}
\frac{\pa_x z}{z} + \frac{\pa_y w}{w} = 0
\end{equation}
The imaginary part gives the identity $ \partial_y \pa_xh = \partial_x \pa_y h$.
\; \\ \; \\
\noindent
The following theorem is essentially proven in \cite{KO} that
\begin{theorem}\cite{KO} Equations (\ref{eq:dimel}) and (\ref{eq:defzw}), for some choice of branch, together imply $ P(z,w) = 0 $. Conversely, if $ z $ and $ w $ satisfy $ P(z,w) =0 $ and (\ref{eq:elzw}), then $ \pa_xh = -\frac{1}{i \pi} \arg(w) $ and $ \pa_yh = \frac{1}{i \pi} \arg(z) $ satisfy the Euler-Lagrange equation (\ref{eq:dimel}).
\end{theorem}

Let $z$ and $ w $ be as in (\ref{eq:defzw}). We can rewrite the Euler-Lagrange equations as:
\begin{equation} \label{eq:burgs}
\pa_x z = F(z) \; \pa_y z \ \ \ \ F(z):= \frac{z}{w} \frac{\pa_z P \left(z,w \right)}{ \pa_w P \left(z,w \right) }
\end{equation}
The equations of motion for $ F(z) $ itself is the complex  Burgers' equation.
\begin{align}
\pa_xF = F \; \pa_yF
\end{align}

\subsubsection{Hamiltonian Framework and Holomorphic Factorization}
The Hamiltonian framework for the dimer models is completely parallel to the
one for the 6-vertex Model.
The Hamiltonian is
\[
H=\int_0^L \tau(p, \pa_y h) dy
\]
where $p$ is the momentum (Legendre dual to $\pa_x h$) and
\begin{equation} \label{ham-dim}
\tau(p,t) = \max_{s}\left(  s \; p - \sigma(s,t) \right)
\end{equation}

The equation of motion (\ref{eq:burgs}) is closely related to the Hamilton's equations in the Hamiltonian framework outlined in Section \ref{Ham-6v-section}. Recall that the Legendre transform of the surface tension induces an identification of $ TM $ and $ T^*M $, in particular, it maps $\pa_xh$ to $ p =\pa_1 \sigma(\pa_xh,\pa_y h)$. Via this identification, define
the image of $z$ of equation (\ref{eq:defzw}) in cotangent bundle to the space of height functions as $ z(y) = \exp( p(y) + i \pi \pa_y h(y)) $ and:
\begin{align*}
l(y) = \log(z(y)) =  p(y) + i \pi \pa_y h(y)
\end{align*}
In terms of $ l$, the canonical Poisson bracket can be written:
\begin{equation}\label{Pb}
\begin{aligned}
\{ l(y), l(y') \} &= 2\pi i \delta'(y-y') \\
\{ l(y), \bar{l}(y') \} &= 0 \\
\{ \bar{l}(y), \bar{l}(y') \} &= - 2\pi i \delta'(y-y')
\end{aligned}
\end{equation}
Note that the Poisson algebra of functionals in $l$ has Casimir elements. Its Poisson center is generated
by $\int_0^Ll(y)dy$ and $\int_0^L\bar{l}(y)dy$

The equations of motion for $ l $ and $ \bar{l} $ are:
\begin{equation} \label{eq-mo-hol}
\begin{aligned}
\pa_xl &= F(e^l) \pa_yl  \; \; \; \; 
\pa_x\bar{l} &= F(e^{\bar{l}}) \pa_y\bar{l}
\end{aligned}
\end{equation}
where $ F $ was defined in equation (\ref{eq:burgs}).

Now let us rewrite the Hamiltonian (\ref{ham-dim}) in terms of $ l $ and $ \bar{l} $.
\begin{proposition} The Hamiltonian density satisfies the following differential equations
\begin{equation} \label{eq:burgham}
\frac{\pa^2 \H}{ \pa l^2} = 2\pi i \; F(e^l), \ \ \frac{\pa^2 \H}{ \pa \bar{l}^2} =-2\pi i \; F(e^{\bar{l}}), \ \
\frac{\pa^2 \H}{ \pa l \pa \bar{l}} = 0
\end{equation}
\end{proposition}

\begin{proof}
Hamiltonian equations of motion with the Poisson brackets (\ref{Pb})
are:
\[
\pa_x l(x,y)=\{ \H, l(y)\}(x,y)=\int_0^L\frac{\pa \H}{\pa \bar{l}}(x,y') \; 2\pi i \delta'(y'-y)dy'=
2 \pi i \left( \frac{\pa^2 \H}{ \pa l^2}\pa_yl+ \frac{\pa^2 \H}{ \pa \bar{l}^2}\pa_y\bar{l} \right)
\]
and a similar equation for $\bar{l}$. Comparing with equations of motion we have
equation (\ref{eq:burgham}).
\end{proof}

\begin{corollary}
This implies
\[
\mathcal{H}(p,\pa_y h)=G(l)+G(\bar{l})+Al+\bar{A}\bar{l}+C
\]
where $G(l)$ is an the antiderivative of $F(l)$. Note that since $\int_0^Ll(y)dy$ and
$\int_0^L\bar{l}(y)dy$ are Casimir functions, the equations of motion for such Hamiltonian
do not depend on $A$ and $C$.
\end{corollary}

It is clear that this Hamiltonian Poisson commute with
functionals
\begin{align*}
I_n = \int_0^L l^n(y) dy, \ \ \bar{I}_n = \int_0^L \bar{l}^n(y) dy
\end{align*}
form a commuting family and $I_n$ and $\bar{I}_n$ mutually Poisson commute.

\subsection{The Homogeneous Dimer Model on a Hexagonal Lattice} \label{subsec:hexdimer}

In this case $\phi_\epsilon(\mathcal{L})$ is the hexagonal lattice embedded into
a cylinder with all edge weights being $1$. The existence of these integrals of course is
due to the integrability of the Burgers equation.

\subsubsection{Surface Tension}
The spectral curve for the hexagonal dimer model is:
\begin{align*}
P(z,w) = 1 - z - w
\end{align*}
Three terms correspond to three dimer configurations on the fundamental
domain for a hexagonal lattice with periodic boundary conditions, see Figure \ref{fig:fund-hex-P}.
\begin{figure}[h]
\begin{tikzpicture}[scale=1]

\coordinate (d0) at (.25,-.25);
\coordinate (d1) at (-.25,.25);

\draw [ultra thick] (d0) -- (d1);
\draw (d0) -- ( $ (d0) + (.5,0) $);
\draw (d0) -- ( $ (d0) - (0,.5) $);
\draw (d1) -- ( $ (d1) - (.5,0) $);
\draw (d1) -- ( $ (d1)  + (0,.5) $);

\node at (.5,-.5){\tiny{0}};
\node at (-.25,-.25){\tiny{0}};
\node at (.25,.25){\tiny{0}};
\node at (-.5,.5){\tiny{0}};

\fill (d0) circle (2pt);
\fill[white] (d1) circle (2pt);\draw (d1) circle (2pt);

\end{tikzpicture} \;
\begin{tikzpicture}[scale=1]

\draw (d0) -- (d1);
\draw [ultra thick] (d0) -- ( $ (d0) + (.5,0) $);
\draw (d0) -- ( $ (d0) - (0,.5) $);
\draw [ultra thick] (d1) -- ( $ (d1) - (.5,0) $);
\draw (d1) -- ( $ (d1) + (0,.5) $);

\node at (.5,-.5){\tiny{0}};
\node at (-.25,-.25){\tiny{0}};
\node at (.25,.25){\tiny{1}};
\node at (-.5,.5){\tiny{1}};

\fill (d0) circle (2pt);
\fill[white] (d1) circle (2pt);\draw (d1) circle (2pt);

\end{tikzpicture} \;
\begin{tikzpicture}[scale=1]

\draw (d0) -- (d1);
\draw (d0) -- ( $ (d0) + (.5,0) $);
\draw [ultra thick] (d0) -- ( $ (d0) - (0,.5) $);
\draw (d1) -- ( $ (d1) - (.5,0) $);
\draw [ultra thick] (d1) -- ( $ (d1) + (0,.5) $);

\node at (.5,-.5){\tiny{1}};
\node at (-.25,-.25){\tiny{0}};
\node at (.25,.25){\tiny{1}};
\node at (-.5,.5){\tiny{0}};

\fill (d0) circle (2pt);
\fill[white] (d1) circle (2pt);\draw (d1) circle (2pt);

\end{tikzpicture}
\caption{The three states on $ \Gamma_0 $ with periodic boundary conditions for the hexagonal lattice.}   \label{fig:fund-hex-P}
\end{figure}
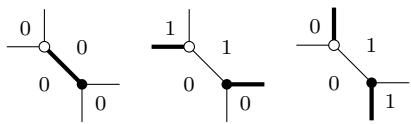 
The surface tension is given in \cite{Ken09} as:
\begin{align*}
\sigma(s,t) &= -\frac{1}{\pi} \Big( L\left( \pi s \right) + L \left( \pi t \right) + L\left( \pi (1 - s - t ) \right) \Big)
\end{align*}
where $ L $ is the Lobavchesky function:
\begin{align*}
L(x) &= - \int_0^x \log \left(2 \sin \left( x \right) \right) dx \\
&= \frac{1}{2}\Im \left( \text{Li}_2( e^{2 i x} ) \right)
\end{align*}
where $ \text{Li}_2(z) = \sum_{k=1} \frac{z^k}{k^2} $ is the dilogarithm.

The partial derivatives are:
\begin{equation}\label{eq:parsigmafv}
\begin{aligned}
\pa_s\sigma( s,t) &=  \log \left( \frac{\sin(\pi s)}{\sin \left( \pi (s+ t) \right) }  \right) \\
\pa_t\sigma(s,t) &= \log \left( \frac{\sin(\pi t)}{\sin \left( \pi ( s+t  ) \right) }  \right)
\end{aligned}
\end{equation}

The Euler-Lagrange equations are:
\begin{align} \label{eq:hexpde}
\pa_x^2h \; \frac{\sin(\pi \pa_yh ) }{ \sin(\pi \pa_xh) } - 2 \; \pa_x\pa_yh \;  \cos \left( \pi ( \pa_y h +  \pa_x h) \right) + \pa_y^2h \; \frac{\sin(\pi \pa_xh)}{\sin(\pi \pa_yh )} =  0
\end{align}

Solving the equation of the spectral curve for $w$, we obtain $F(z)  = z/(1-z)$.
In terms of $l$-variable the Euler-Lagrange  equation become:
\[
\pa_x l=-\frac{e^l}{1-e^l} \; \pa_y l
\]

\subsubsection{Hamiltonian framework} The Legendre transform of the dimer action
for this particular case for the homogeneous hexagonal lattice gives:
\begin{align}\label{hom-hex-ham}
\H = \frac{1}{ 2 \pi i}\; \left(\text{Li}_2( e^{l} ) - \text{Li}_2( e^{\bar{l}}) \right) 
\end{align}
where $l = p + i \pi h_y$ and $ \text{Li}_2(x)=\sum_{k=1}^\infty \frac{z^k}{k^2}$.

Indeed, minimizing $\tau(p, \pa_y h) =  p \; \pa_xh - \sigma( \pa_xh,  \pa_y h )$
in $\pa_xh$ we obtain the equation for the critical point:
\begin{align}
p(\pa_x h,\pa_y h) = \log \left( \frac{\sin(\pi \pa_xh(x,y))}{\sin \left( \pi ( \pa_y h(x,y) + \pa_x h(x,y) \right) }  \right)
\end{align}
which is easy to solve for $\pa_xh(x,y)$:
\begin{align}
\pa_xh(x,y) = \frac{i}{2 \pi} \log \left(  \frac{1 - e^{p(x,y) + i \pi \pa_yh(x,y)}}{1 - e^{p(x,y) - i \pi \pa_yh(x,y)} } \right)
\end{align}
This immediately gives (\ref{hom-hex-ham}). It is straightforward to check that equations (\ref{eq:burgham}) hold.

\section{Free Fermionic point}\label{ff-point}
\subsection{Free fermionic point $\Delta=0$} The six vertex model is related to the dimer model found by replacing each vertex of the six vertex model with the medial graph:
\\
\begin{figure}[h]

\begin{tikzpicture}[scale=1]
\coordinate (d0) at (1,0);
\coordinate (d1) at (0,1);
\coordinate (d2) at (-1,0);
\coordinate (d3) at (0,-1);

\coordinate (d4) at (.5,-.5);
\coordinate (d5) at (-.5,.5);

\coordinate (e0) at (2,0);
\coordinate (e1) at (0,2);
\coordinate (e2) at (-2,0);
\coordinate (e3) at (0,-2);
\node at (.65,.65){\footnotesize{$\beta_1$}};
\node at (.1,.1){\footnotesize{$\gamma$}};
\node at (-.65,-.65){\footnotesize{$\beta_2$}};
\node at (-.4,.9){\footnotesize{$\alpha_1$}};
\node at (-.9,.4){\footnotesize{$\alpha_2$}};
\node at (.4,-.9){\footnotesize{$\alpha_3$}};
\node at (.9,-.4){\footnotesize{$\alpha_4$}};
\draw plot coordinates {(d0)(d1)(d2)(d3)(d0)};
\draw plot coordinates {(d4)(d5)};
\draw [decoration={markings,mark=at position 1 with {\arrow[scale=2]{>}}},
    postaction={decorate}] plot coordinates {(d0)(e0)};
\draw[decoration={markings,mark=at position 1 with {\arrow[scale=2]{>}}},
    postaction={decorate}]  plot coordinates {(d1)(e1)};
\draw plot coordinates {(d2)(e2)};
\draw plot coordinates {(d3)(e3)};

\fill (d0) circle (2pt);
\fill[white] (d1) circle (2pt);\draw (d1) circle (2pt);
\fill (d5) circle (2pt);
\fill[white] (d2) circle (2pt);\draw (d2) circle (2pt);
\fill (d3) circle (2pt);
\fill[white] (d4) circle (2pt);\draw (d4) circle (2pt);

\end{tikzpicture}

\end{figure} 

The correspondence between six vertex states and dimer configurations is as illustrated in Figure \ref{fig:dc} of Appendix \ref{h-f}.

The dimer model edge weights are related to the six vertex weights by:
\begin{align*}
\alpha = \sqrt{b} \hspace{40 pt}
\beta = \frac{c-b}{a} \hspace{40 pt}
\gamma = a
\end{align*}
Parameterize the six vertex weights as:
\begin{align*}
a = \cos(u) \hspace{40 pt}
b = \sin(u) \hspace{40 pt}
c = 1
\end{align*}

The characteristic polynomial, given in (\ref{eq:dimercityp}) simplifies to:
\begin{align*}
P(z,w) = (w z - 1) \cos(u) + (z + w) \sin(u)
\end{align*}

The spectral curve $P(z,w)=0$ gives
\[
w=\frac{z \cos u - \sin u}{ \cos u + z \sin u}
\]
and thus,
\[
F(z)=-\frac{z}{(z\cos u + \sin u)(z\sin u - \cos u)}
\]

\subsubsection{The Surface Tension and Equations of Motion}
The formula for the free energy on the torus derived from the Kasteleyn solution is:
\begin{align*}
f(H,V)  = \frac{1}{ ( 2 \pi i )^2} \int_{|z| = e^{V}} \int_{|w| = e^{H}} \log \big( (w z - 1) \cos(u) + (z + w) \sin(u) \big) \frac{dz}{z} \frac{dw}{w}
\end{align*}
It is straightforward to compute the partial derivatives of the free energy:
\begin{align*}
\partial_H f &= \frac{1}{\pi} \cos^{-1} \left( \frac{\sinh(V-H) \tan(u) - \sinh(V+H) \cot(u) }{2 \cosh(H) } \right) \\
\partial_V f &= \frac{1}{\pi} \cos^{-1} \left( \frac{\sinh(H-V) \tan(u) - \sinh(V+H) \cot(u) }{2 \cosh(V) } \right)
\end{align*}

This defines the map $\RR^2\to \RR^2, \ (H,V)\mapsto (\pa_Hf, \pa_Vf)$. Partial derivatives of $\sigma$
define the inverse map $(s,t)\mapsto (\pa_s \sigma, \pa_t \sigma)$. It has the following explicit form:
\begin{equation}\label{eq:parsigmaff}
\begin{aligned} 
\partial_s \sigma(s,t) &= -\sinh^{-1}\left(\frac{ \sin (\pi t) \cos (\pi s) -\cos (2 u) \cos (\pi  t) \sin( \pi s) }{\sin(2 u) \sin(\pi s) } \right)  \\
\partial_t \sigma(s,t) &= -\sinh^{-1}\left( \frac{  \sin (\pi s) \cos (\pi t) -\cos (2 u) \cos (\pi  t) ) \sin( \pi t)}{ \sin(2 u) \sin(\pi t) } \right)
\end{aligned}
\end{equation}

The Euler-Lagrange equations in the regions where $\pa_xh, \pa_yh\in (0,1) $ are:
\begin{equation}\label{eq:ffpde}
\begin{aligned}
0 = & \frac{\sin(\pi \pa_y h) }{\sin(\pi \pa_xh) } \; \pa^2_xh \\ & - 2 \;  \big(  \cos(\pi \pa_x h) \cos( \pi \pa_y h) + \cos(2 u)  \sin( \pi \pa_x h) \sin(\pi pa_y h) \big) \; \pa_x \pa_yh  \\ &+  \frac{\sin(\pi \pa_xh) }{\sin(\pi \pa_y h)} \; \pa^2_yh 
\end{aligned}
\end{equation}

In the five vertex limit $ u \rightarrow \pi/2 $ (see Appendix \ref{sec:fivevertex} for details), we recover the PDE for the limit shape of the hexagonal dimer model (\ref{eq:hexpde}). Recall the dimer height function differs from the six vertex height function by a linear term; after shifting $ \pa_x h \rightarrow \pa_x h - \frac{1}{2} $ and $ \pa_y h \rightarrow \pa_y h - \frac{1}{2} $, equation (\ref{eq:ffpde}) matches with the PDE derived from Bethe ansatz, see \cite{PR}, \cite{ZJ}.

\subsubsection{The Hamiltonian Framework}
The Hamiltonian for the free fermionic six vertex model is most easily computed from the Bethe ansatz \cite{PR}, \cite{NK}:
\begin{align*}
\mathcal{H} = \frac{1}{2 \pi i} \Big( \big( \text{Li}_2( e^l \tan(u) )  -  \text{Li}_2( -  e^l \cot(u)) \big) -  \big( \text{Li}_2( e^{\bar{l}} \tan(u) )  -  \text{Li}_2( -  e^{\bar{l}} \cot(u)) \big) \Big)
\end{align*}
It is straightforward to check that it satisfies (\ref{eq:burgham}).

We can recover the Hamiltonian (\ref{hom-hex-ham}) of the hexagonal dimer model by first translating the momentum
\begin{align*}
p \rightarrow p + \log( \sin(2 u ) )
\end{align*}
and then taking the 5-vertex limit $ u \rightarrow \frac{\pi}{2} $. The translation corresponds to difference of conjugate momenta $ p = \partial_s \sigma(s,t) $ of the hexagonal dimer model (\ref{eq:parsigmafv}) and free fermionic six vertex model (\ref{eq:parsigmaff}) in the limit $ u \rightarrow \frac{\pi}{2} $.

\section*{Conclusion}
In this paper, we presented infinitely many integrals of motion for the limit shape equations of the 6-vertex model.
We conjecture that it defines an integrable system, but this remains to be proven. A proof of integrability would require a transformation to action-angle variables or a Lax pair presentation.

One should notice the analogy between Conformal Field Theory and the free fermionic limit of the six vertex model. In conformal field
theory (for example in the Gaussian field theory), the integrals of motion factorize into the sum of holomorphic
and antiholomorphic part. When the CFT is perturbed to an integrable QFT \cite{Zam} , the integrals of motion no longer factorize. We observe the same phenomena in the 6-vertex model. When $\Delta=0$
the Hamiltonian factorizes into the holomorphic and antiholomorphic part in terms of Burgers
variables. But, already terms linear in $\Delta$ do not have such structure. This observation goes in
line with the discrete holomorphicity of Dimer models (see for example \cite{Sm} and references therein)
and suggests that such structure is not fully present in the 6-vertex model, which agrees with
\cite{ZJ}.

Our results suggest that the equation describing limit shapes in the 6-vertex model (and in Dimer models)
is an example of an integrable PDE with gradient constraints. In our proof of Poisson commutativity
of the integrals of motion, we assumed that the height function is away from gradient constraints. It is well known that imposing
boundary conditions in integrable PDEs destroy the integrability, if the boundary conditions are not of a very special type. It would be interesting to study which sort of gradient constraints in an integrable PDE destroy integrability and which
ones preserve it.

\appendix

\section{The partial Legendre transform of $\sigma(s,t)$}\label{pTt}
Let $\tau$ be the Legendre transform of $\sigma$ with respect to the first argument:
\begin{align*}
\tau( p, \xi) = \max_{\nu \in \mathbb{R} } \big( p \; \nu - \sigma(\nu, \xi) \big)
\end{align*}
Assuming that $\sigma$ is strictly convex and smooth except, possibly on codimension greater then one
strata, for generic $(p,\xi)$, the maximum is achieved at $ \tilde{\nu}^*(p, \xi) $ defined by:
\begin{align} \label{eq:nustar}
p = \partial_1 \sigma(\nu^*, \xi )
\end{align}
Recall that for a function $f$ in two variables we denote by $\pa_1 f$ and $\pa_2 f$ partial derivatives of $f$
in the first and the second variable respectively.
Then:
\begin{align*}
\tau(p, \xi) = p \; \nu^*(p, \xi ) - \sigma\big( \nu^*(p, \xi), \xi \big)
\end{align*}

A simple but useful general fact about the Legendre transform is:
\begin{lemma}\label{Hess-Lemma} The following identity holds:
\begin{align*}
\frac{\partial_2^2 \tau(p, \xi) }{\partial_1^2 \tau(p, \xi)  } = \partial_1^2 \sigma(\nu^*,\xi) \; \partial_2^2 \sigma(\nu^*,\xi) - \big( \partial_1 \partial_2 \sigma(\nu^*,\xi) \big)^2
\end{align*}
 
\end{lemma}
\begin{proof}
We first compute:
\begin{equation*}
\partial_1 \tau = \partial_{p} \left( p \nu^* - \sigma(\nu^*, \xi) \right)
= \nu^* + p \; \partial_p \nu^* - \partial_1 \sigma(\nu^*, \xi) \; \partial_p \nu^*
= \nu^*
\end{equation*}
We used equation (\ref{eq:nustar}). Differentiation this with respect to $ p $ gives $ 1 = \partial_1^2 \sigma(\nu^*, \xi) \; \partial_1 \nu^* $. So:
\begin{align} \label{eq:p11t}
\partial_1^2 \tau = \frac{1}{\partial_1^2 \sigma }
\end{align}
Similarly:
\begin{equation*}
\partial_2 \tau = \partial_\xi  \left( p \nu^*) - \sigma(\nu^*, \xi) \right) \notag
= - \partial_2 \sigma
\end{equation*}
\begin{equation}\label{eq:p22t}
\partial_2^2 \tau = - \partial_2^2 \sigma - \partial_1 \partial_2 \sigma \; \partial_2 \nu^*
\end{equation}
Differentiation (\ref{eq:nustar}) with respect to $ \xi $ gives:
\begin{align*}
\partial_2 \nu^* \; \partial_1^2 \sigma + \partial_1 \partial_ 2 \sigma = 0
\end{align*}
Solving for $ \partial_2 \nu^* $ and substituting into (\ref{eq:p22t}) we obtain:
\begin{align}
\partial_2 \tau = -\partial_2^2 \sigma + \frac{(\partial_1 \partial_ 2 \sigma)^2}{\partial_1^2 \sigma}
\end{align}
Combining with (\ref{eq:p11t}) we arrive to:
\begin{align*}
\frac{\partial_2^2 \tau}{\partial_1^2 \tau} = - \partial_1^2 \sigma \; \partial_2^2 \sigma + ( \partial_1 \partial_2 \sigma)^2
\end{align*}
\end{proof}

\section{Dimer City and Six Vertex Model} \label{h-f}
The six-vertex model at the free- fermion point is equivalent to a dimer model with the dimer city \cite{MW} placed at each vertex of the square lattice:

\begin{figure}[h]

\begin{tikzpicture}[scale=1]
\coordinate (d0) at (1,0);
\coordinate (d1) at (0,1);
\coordinate (d2) at (-1,0);
\coordinate (d3) at (0,-1);

\coordinate (d4) at (.5,-.5);
\coordinate (d5) at (-.5,.5);

\coordinate (e0) at (2,0);
\coordinate (e1) at (0,2);
\coordinate (e2) at (-2,0);
\coordinate (e3) at (0,-2);
\node at (.65,.65){\footnotesize{$\beta_1$}};
\node at (.1,.1){\footnotesize{$\gamma$}};
\node at (-.65,-.65){\footnotesize{$\beta_2$}};
\node at (-.4,.9){\footnotesize{$\alpha_1$}};
\node at (-.9,.4){\footnotesize{$\alpha_2$}};
\node at (.4,-.9){\footnotesize{$\alpha_3$}};
\node at (.9,-.4){\footnotesize{$\alpha_4$}};
\draw plot coordinates {(d0)(d1)(d2)(d3)(d0)};
\draw plot coordinates {(d4)(d5)};
\draw [decoration={markings,mark=at position 1 with {\arrow[scale=2]{>}}},
    postaction={decorate}] plot coordinates {(d0)(e0)};
\draw[decoration={markings,mark=at position 1 with {\arrow[scale=2]{>}}},
    postaction={decorate}]  plot coordinates {(d1)(e1)};
\draw plot coordinates {(d2)(e2)};
\draw plot coordinates {(d3)(e3)};

\fill (d0) circle (2pt);
\fill[white] (d1) circle (2pt);\draw (d1) circle (2pt);
\fill (d5) circle (2pt);
\fill[white] (d2) circle (2pt);\draw (d2) circle (2pt);
\fill (d3) circle (2pt);
\fill[white] (d4) circle (2pt);\draw (d4) circle (2pt);

\end{tikzpicture}

\end{figure} 

The correspondence of dimer configurations and six vertex states is as shown in Figure \ref{fig:dc}.

\input{dimercityheights}

In terms of the dimer weights, the weights of the six-vertex configurations are:
\begin{align*}
w_1 &= \beta_1 \beta_2 \gamma + \beta_2 \alpha_2 \alpha_3 + \beta_1 \alpha_1 \alpha_4 \\
w_2 &=  \gamma \\
w_3 &= \alpha_1 \alpha_3 \\
w_4 &= \alpha_4 \alpha_2 \\
w_5 &= \alpha_2  \alpha_3 + \beta_1 \gamma \\
w_6 &= \alpha_1 \alpha_4 + \beta_2
\end{align*}
It is straightforward to check that for all choice of parameters, $ \Delta = 0 $.

The height function for the dimer model was explained in Section (\ref{dimer-height}). For the dimer city, we have chosen the first $ w_1 $ type vertex as the reference dimer configuration. The height changes are as shown in Figure \ref{fig:dc}.

The characteristic polynomial for the dimer city is computed from dimer configurations satisfying toric boundary conditions. These are the first 6 dimer configurations shown in Figure \ref{fig:dc}, giving:
\begin{align} \label{eq:dimercityp}
P(z,w) = \left(\beta_1 \beta_2 \gamma + \alpha_1 \alpha_4 \beta_1 + \alpha_2 \alpha_3 \beta_2 \right) - (\alpha_1 \alpha_3) w - (\alpha_2 \alpha_4) z^{-1} - \gamma w z^{-1}
\end{align}

Denote by $ \theta^\text{dimer} $ the height function of the dimer model, and by $ \theta^\text{6v} $ the height function for the six vertex model. It is clear that:
\begin{align*}
\theta^\text{dimer} =  \theta^\text{6v} + x/2 + y/2
\end{align*}

\section{$5$-vertex limits of the $6$-vertex model}\label{sec:fivevertex}
The 5-vertex model is the limit of the 6-vertex model with the weight $ w_2 = 0 $. The $ R $-matrix for this model is:
\begin{align*}
R = \begin{pmatrix}
a e^{H+V} & 0 & 0 & 0 \\
0 & b e^{H-V} & c & 0 \\
0 & c & b e^{V-H} & 0 \\
0 & 0 & 0 & 0
\end{pmatrix}
\end{align*}

Other $5$-vertex models are found by setting $ w_3 $, $ w_4 $, or $ w_1 $ each to zero. They are related to model with $ w_1 = 0 $ by the symmetries that interchange thin vertical edges with thick vertical edges, or thin horizontal edges with thick horizontal edges, or both.

The 5-vertex model an be obtained from the 6-vertex model in the limit that the magnetic fields go to infinity as follows:
\begin{enumerate}
\item If $ b^2 - c^2 > 0 $: using the parametrization with $ (a,b,c) = ( \sinh \big( u - \gamma \big), \sinh(u), \sinh(\gamma) )$, change variables to $ (\xi, l, m) $ as:
\begin{align*}
H =  \frac{\gamma}{2} + l \hspace{20pt} V =  \frac{\gamma}{2} + m  \hspace{20pt} u = \gamma + \xi 
\end{align*}
Taking the limit $ \gamma \rightarrow \infty $ the $ R$-matrix converges (up to a constant factor) to:
\begin{align*}
R = \begin{pmatrix}
2  \; \sinh(\xi) e^{l+m}& 0 & 0 & 0 \\
0 & e^{\xi - l + m} & 1 & 0 \\
0 & 1 &  e^{\xi + l-m} & 0 \\
0 & 0 & 0 & 0
\end{pmatrix}
\end{align*}

\item If $ b^2 - c^2 = 0 $: using the parametrization for $ (a,b,c) = ( r \sin( \gamma - u), r \sin(u), r \sin( \gamma) ) $, change variables to:
\begin{align*}
H =  -\frac{1}{2} \ln(\gamma - u) + l \hspace{20pt} V = -\frac{1}{2} \ln(\gamma - u) + m 
\end{align*}
then taking the limit $ \gamma \rightarrow u $, the $ R $-matrix converges to:
\begin{align*}
R = \begin{pmatrix}
 e^{l+m}& 0 & 0 & 0 \\
0 & \sin(u) e^{l-m}& \sin(u) & 0 \\
0 & \sin(u) &  \sin(u) e^{-l + m } & 0 \\
0 & 0 & 0 & 0
\end{pmatrix}
\end{align*}

\item If $ b^2 - c^2 < 0 $, using the parametrization with $ (a,b,c) = ( \sinh( \gamma - u ), \sinh(u), \sinh(\gamma) ) $, change variables to:
\begin{align*}
H =  \frac{\gamma}{2} + l \hspace{20pt} V =  \frac{\gamma}{2} + m  \hspace{20pt} u = \gamma - \xi 
\end{align*}
Taking the limit $ \gamma \rightarrow \infty $ the $ R$-matrix converges (up to a constant factor) to:
\begin{align*}
R = \begin{pmatrix}
2 \sinh(\xi) e^{H+V}& 0 & 0 & 0 \\
0 & e^{- \xi + H-V} & 1 & 0 \\
0 &  1 &  e^{- \xi + V-H} & 0 \\
0 & 0 & 0 & 0
\end{pmatrix}
\end{align*}
\end{enumerate}

The 5-vertex model is equivalent to the dimer model with interactions on the hexagonal lattice \cite{HWKK}. At the free fermion point, $ b^2 - c^2 = 0 $, the five-vertex model is equivalent to the pure dimer model on the hexagonal lattice found by replacing each vertex of the square lattice with the graph:
\begin{figure}[h] \label{fig:hexagonalfivevertex}
\begin{tikzpicture}[scale=1]

\coordinate (d0) at (-.4,-.4);
\coordinate (d1) at ( .4, .4);

\draw (d0) -- (d1);
\draw (d0) -- ( $ (d0) - (.75,0) $);
\draw (d0) -- ( $ (d0) - (0,.75) $);
\draw (d1) -- ( $ (d1) + (.75,0) $);
\draw (d1) -- ( $ (d1)  + (0,.75) $);

\node at (.1,-.1) {\footnotesize{ $ \alpha $ } };
\node at (-1.4,-.5){\footnotesize{ $ \beta$ } };
\node at (1.4,.5){\footnotesize{ $ \beta $ } };
\node at (-.5,-1.4){\footnotesize{ $ \gamma $}};
\node at (.5, 1.4) {\footnotesize{ $ \gamma $ }};

\fill (d0) circle (2pt);
\fill[white] (d1) circle (2pt);\draw (d1) circle (2pt);

\end{tikzpicture}
\end{figure} 

\noindent where the weights are:
\begin{align*}
\alpha &= e^{l+m}\hspace{20pt} \beta = e^\frac{l-m}{2} \hspace{20pt} \gamma = e^\frac{m-l}{2}
\end{align*}

\noindent The correspondence of five vertex states and dimer configurations is as shown in Figure (\ref{fig:fiveVertex}).
\begin{figure}[h] 
\begin{centering}
\begin{tikzpicture}[baseline, scale=1]
\coordinate (d0) at (-.25,-.25);
\coordinate (d1) at (.25,.25);
\draw [ultra thick] (d0) -- (d1);
\draw (d0) -- ( $ (d0) - (.5,0) $);
\draw (d0) -- ( $ (d0) - (0,.5) $);
\draw (d1) -- ( $ (d1) + (.5,0) $);
\draw (d1) -- ( $ (d1) + (0,.5) $);

\fill (d0) circle (2pt);
\fill[white] (d1) circle (2pt);\draw (d1) circle (2pt);

\node at (0,-1){$ \alpha $};
\end{tikzpicture} \; \; \; \begin{tikzpicture}[baseline]
\draw (-.8,0) -- (.8,0);
\draw (0,-.8) -- (0,.8);
\node at (0,-1){$e^{l+m}$};
\end{tikzpicture} 

\begin{tikzpicture}[baseline, scale=1]
\draw (d0) -- (d1);
\draw [ultra thick] (d0) -- ( $ (d0) - (.5,0) $);
\draw (d0) -- ( $ (d0) - (0,.5) $);
\draw [ultra thick] (d1) -- ( $ (d1) + (.5,0) $);
\draw (d1) -- ( $ (d1) + (0,.5) $);

\fill (d0) circle (2pt);
\fill[white] (d1) circle (2pt);\draw (d1) circle (2pt);

\node at (0,-1){$ \beta^2 $};

\end{tikzpicture} \; \; \; \begin{tikzpicture}[baseline]
\draw (-.8,0) -- (.8,0);
\draw (0,-.8) -- (0,.8);
\draw[ultra thick](-.8,0)--(.8,0);
\node at (0,-1){$c \; e^{l-m}$};
\end{tikzpicture}

\begin{tikzpicture}[baseline, scale=1]
\draw (d0) -- (d1);
\draw (d0) -- ( $ (d0) - (.5,0) $);
\draw [ultra thick] (d0) -- ( $ (d0) - (0,.5) $);
\draw (d1) -- ( $ (d1) + (.5,0) $);
\draw [ultra thick] (d1) -- ( $ (d1) + (0,.5) $);

\fill (d0) circle (2pt);
\fill[white] (d1) circle (2pt);\draw (d1) circle (2pt);

\node at (0,-1){$ \gamma^2 $};

\end{tikzpicture} \; \; \;\begin{tikzpicture}[baseline]
\draw (-.8,0) -- (.8,0);
\draw (0,-.8) -- (0,.8);
\draw[ultra thick](0,-.8)--(0,.8);
\node at (0,-1){$c \; e^{m-l}$};
\end{tikzpicture}

\begin{tikzpicture}[baseline, scale=1]

\draw (d0) -- (d1);
\draw (d0) -- ( $ (d0) - (.5,0) $);
\draw  [ultra thick](d0) -- ( $ (d0) - (0,.5) $);
\draw [ultra thick] (d1) -- ( $ (d1) + (.5,0) $);
\draw (d1) -- ( $ (d1) + (0,.5) $);

\fill (d0) circle (2pt);
\fill[white] (d1) circle (2pt);\draw (d1) circle (2pt);

\node at (0,-1){$ \gamma \beta $};

\end{tikzpicture} \; \; \; \begin{tikzpicture}[baseline]
\draw (-.8,0) -- (.8,0);
\draw (0,-.8) -- (0,.8);
\draw[ultra thick](0,-.8)--(0,0)--(.8,0);
\node at (0,-1){$c$};
\end{tikzpicture}

\begin{tikzpicture}[baseline, scale=1]
\draw (d0) -- (d1);
\draw  [ultra thick](d0) -- ( $ (d0) - (.5,0) $);
\draw (d0) -- ( $ (d0) - (0,.5) $);
\draw  (d1) -- ( $ (d1) + (.5,0) $);
\draw [ultra thick](d1) -- ( $ (d1) + (0,.5) $);

\fill (d0) circle (2pt);
\fill[white] (d1) circle (2pt);\draw (d1) circle (2pt);

\node at (0,-1){$ \gamma \beta $};
\end{tikzpicture} \; \; \; \begin{tikzpicture}[baseline]
\draw (-.8,0) -- (.8,0);
\draw (0,-.8) -- (0,.8);
\draw[ultra thick](-.8,0)--(0,0)--(0,.8);
\node at (0,-1){$c$};
\end{tikzpicture}

\end{centering}
\caption{The $ 5 $-vertex configurations and weights, and the corresponding hexagonal dimer model configurations and weights.} \label{fig:fiveVertex}
\end{figure}
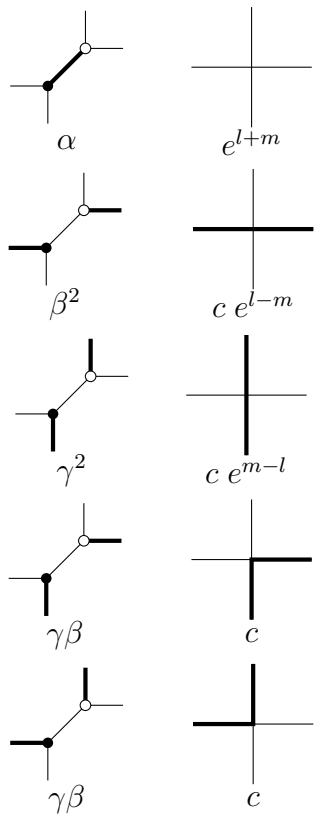 

 It is straightforward to check that as before, the height function of the dimer model and $ 5 $-vertex model are related by $ \theta^\text{dimer} =  \theta^\text{6v} + x/2 + y/2 $.

Thus, the hexagonal dimer model is the limit of the free fermionic 6-vertex model in the limit $  u \rightarrow \gamma $ and $ H, V \rightarrow \infty $. Indeed, by taking the limit $ u \rightarrow 0 $ in the PDE (\ref{eq:ffpde}) of the limit shape of the free fermionic six vertex model, we recover the PDE (\ref{eq:hexpde}) for the hexagonal dimer model.

\section{Dimer Height Functions}
In Section (\ref{dimer-height}), we described the relative height function of two dimer configurations. A more general construction is as follows.

Let $ C_0 $ be a cellular one-chain, satisfying
\begin{align} \label{eq:chaindimer}
\partial C_0 = \sum_{ v \in B } v - \sum_{v \in W } v
\end{align}
where $ B $ and $ W $ are the sets of black and white vertices of $ \Gamma $. Note that in particular, any dimer configuration can be regarded as a one-chain satisfying (\ref{eq:chaindimer}). For any dimer configuration $ \mathcal{D} $, the chain $ \mathcal{D} - C_0 $ satisfies $ \partial (\mathcal{D} - C_0 ) = 0 $, and so the chain $ \mathcal{D} - C_0 $ locally integrates to a function on faces $ \theta_{\mathcal{D}, C_0 } $.

An important class of dimer models are dimers on bipartite 3-valent graphs. The dimer models we considered in this paper in section \ref{subsec:hexdimer} and \ref{ff-point} are both of this type. A different height function, arising from the relation to tiling models, is defined as follows:
\begin{itemize}
\item Orient all edges of the graph from white vertices to back vertices.
\item Using the orientation of the surface define the function $\theta_\mathcal{D} $ on faces
adjacent to an inner edge of $\Gamma$ occupied by a dimer in $ \mathcal{D} $ as it is shown below.
\item The same definition is used to define $\theta_\mathcal{D} $ on faces adjacent to boundary dimer edges.
\end{itemize}
\begin{figure}[h]
\begin{tikzpicture}[scale=1]

\coordinate (d0) at (1,0);
\coordinate (d1) at (-1,0);

\coordinate (t1) at (0.45, 0.779423);
\coordinate (t2) at (0.45, -0.779423);

\draw[thick, decoration={markings,mark=at position 0.6 with {\arrow[scale=2]{>}}}, postaction={decorate}] (d1) -- (d0);
\draw[ultra thick] (d0)--(d1);

\draw (d0) -- ( $(d0) + (t1) $);
\draw (d0) -- ( $(d0) + (t2) $);

\draw (d1) -- ( $(d1) - (t1) $);
\draw (d1) -- ( $(d1) - (t2) $);

\fill (d0) circle (2pt);
\fill[white] (d1) circle (2pt);\draw (d1) circle (2pt);

\node at (-2,0){$\theta_\mathcal{D} + \frac{1}{2}$};
\node at (2,0){$\theta_\mathcal{D} + \frac{1}{2}$};
\node at (0,.6){$\theta_\mathcal{D}+1$};
\node at (0,-.6){$\theta_\mathcal{D}$};

\end{tikzpicture}\label{fig:dimer-h-fncn}
\end{figure} 

In other words, if $ C_0 = \frac{1}{3} \sum_{e \in E(\Gamma) } e $ is the "uniform cover", then
\begin{align*}
\theta_\mathcal{D} = \frac{3}{2} \theta_{\mathcal{D}, \mathcal{C}_0 }
\end{align*}

In this section we prove formulas (\ref{hf-weights}) and (\ref{T}), which give the dimer model weight in terms of the height function.

Given an edge $ e $, oriented from the white vertex to the black vertex, we denote by $ e_L $ and $ e_R $ the faces to the left and right of $ e $.

\begin{lemma} \label{lem:leftright}
Given a dimer configuration $ \mathcal{D} $ with height function $ \theta $, we have:
\begin{align*}
\theta_\mathcal{D}(e_L) - \theta_\mathcal{D}(e_R) +\frac{1}{2} = \begin{cases} 0 &\mbox{if } e \not\in \mathcal{D} \\
\frac{3}{2} & \mbox{if } e \in \mathcal{D} \end{cases}
\end{align*}
\end{lemma}
\begin{proof}
The proof is by checking cases. When $ e \in \mathcal{D} $, the height change is 1 as shown above. When $ e \not \in \mathcal{D}$, there are four cases to consider as shown below.
\begin{figure}[h!]
\begin{tikzpicture}[scale=.65]

\coordinate (d0) at (1,0);
\coordinate (d1) at (-1,0);

\coordinate (t1) at (0.45, 0.779423);
\coordinate (t2) at (0.45, -0.779423);

\draw[ decoration={markings,mark=at position 0.6 with {\arrow[scale=2]{>}}}, postaction={decorate}] (d1) -- (d0);
\draw (d0)--(d1);

\draw [ultra thick] (d0) -- ( $(d0) + (t1) $);
\draw (d0) -- ( $(d0) + (t2) $);

\draw (d1) -- ( $(d1) - (t1) $);
\draw [ultra thick](d1) -- ( $(d1) - (t2) $);

\fill (d0) circle (2pt);
\fill[white] (d1) circle (2pt);\draw (d1) circle (2pt);

\node at (-2,0){\tiny{$ \frac{1}{2}$}};
\node at (2,0){\tiny{$ \frac{1}{2}$}};
\node at (0,.6){\tiny{$ -\frac{1}{2}$}};
\node at (0,-.6){\tiny{$ 0 $}};

\end{tikzpicture} \;
\begin{tikzpicture}[scale=.65]

\coordinate (d0) at (1,0);
\coordinate (d1) at (-1,0);

\coordinate (t1) at (0.45, 0.779423);
\coordinate (t2) at (0.45, -0.779423);

\draw[decoration={markings,mark=at position 0.6 with {\arrow[scale=2]{>}}}, postaction={decorate}] (d1) -- (d0);
\draw (d0)--(d1);

\draw[ultra thick] (d0) -- ( $(d0) + (t1) $);
\draw (d0) -- ( $(d0) + (t2) $);

\draw [ultra thick] (d1) -- ( $(d1) - (t1) $);
\draw (d1) -- ( $(d1) - (t2) $);

\fill (d0) circle (2pt);
\fill[white] (d1) circle (2pt);\draw (d1) circle (2pt);

\node at (-2,0){\tiny{$ - 1$}};
\node at (2,0){\tiny{$ + \frac{1}{2}$}};
\node at (0,.6){\tiny{$  - \frac{1}{2}$}};
\node at (0,-.6){\tiny{$0 $}};
\end{tikzpicture}

\begin{tikzpicture}[scale=.65]

\coordinate (d0) at (1,0);
\coordinate (d1) at (-1,0);

\coordinate (t1) at (0.45, 0.779423);
\coordinate (t2) at (0.45, -0.779423);

\draw[ decoration={markings,mark=at position 0.6 with {\arrow[scale=2]{>}}}, postaction={decorate}] (d1) -- (d0);
\draw (d0)--(d1);

\draw  (d0) -- ( $(d0) + (t1) $);
\draw [ultra thick](d0) -- ( $(d0) + (t2) $);

\draw [ultra thick](d1) -- ( $(d1) - (t1) $);
\draw (d1) -- ( $(d1) - (t2) $);

\fill (d0) circle (2pt);
\fill[white] (d1) circle (2pt);\draw (d1) circle (2pt);

\node at (-2,0){\tiny{$ - \frac{1}{2}$}};
\node at (2,0){\tiny{$ - 1$}};
\node at (0,.6){\tiny{$ -1 $}};
\node at (0,-.6){\tiny{$ 0 $}};

\end{tikzpicture} \;
\begin{tikzpicture}[scale=.65]

\coordinate (d0) at (1,0);
\coordinate (d1) at (-1,0);

\coordinate (t1) at (0.45, 0.779423);
\coordinate (t2) at (0.45, -0.779423);

\draw[decoration={markings,mark=at position 0.6 with {\arrow[scale=2]{>}}}, postaction={decorate}] (d1) -- (d0);
\draw (d0)--(d1);

\draw (d0) -- ( $(d0) + (t1) $);
\draw [ultra thick] (d0) -- ( $(d0) + (t2) $);

\draw  (d1) -- ( $(d1) - (t1) $);
\draw [ultra thick] (d1) -- ( $(d1) - (t2) $);

\fill (d0) circle (2pt);
\fill[white] (d1) circle (2pt);\draw (d1) circle (2pt);

\node at (-2,0){\tiny{ $ \frac{1}{2} $}};
\node at (2,0){\tiny{$ -\frac{3}{2} $}};
\node at (0,.6){\tiny{$- \frac{1}{2}$}};
\node at (0,-.6){\tiny{$0$}};
\end{tikzpicture}
\end{figure}  
\end{proof}
Now we prove:
\begin{lemma} The weight of a dimer configuration is 
\begin{equation}
W(\mathcal{D}) = \prod_{e\in \Gamma} w(e)^{1/3} \; \prod_f q_f^{2 \; \theta_\mathcal{D} (f) / 3}
\end{equation}
where $\theta_\mathcal{D} (f)$ is the value of the height function $\theta_\mathcal{D} $ on the face $f$ and $q_f$ is defined in (\ref{hf-weights}).
\end{lemma}
\begin{proof}
The proof is by computation. Using the definition of $q_f$ we have:
\begin{equation*}
\left( \prod_{e\in \Gamma} w(e)^{1/3} \right) \; \prod_f q_f^{2 \; \theta(f) / 3} \ \
= \left( \prod_{e\in \Gamma} w(e)^{1/3} \right) \prod_f \prod_{e\subset \pa f} w(e)^{2 \; \theta(f) \; \varepsilon(e,f) / 3}
\end{equation*}
Now, we use the  definition of $\varepsilon$,  $ \varepsilon( e, e_L ) = 1 $  and $ \varepsilon( e, e_R ) = -1 $:
\begin{align*}
&= \left( \prod_{e\in \Gamma} w(e)^{1/3} \right) \prod_{e} w(e)^{2 \; \left( \theta(e_L) - \theta(e_R)  \right) /3 } \\
&= \prod_{e} w(e)^{2 \; \left( \theta(e_L) - \theta(e_R) + \frac{1}{2}  \right) /3 } = \prod_{e\in \mathcal{D} } w(e)=W(\mathcal{D})
\end{align*}
Here we used Lemma (\ref{lem:leftright}).
\end{proof}

\end{document}